\documentclass[11pt]{article}
\usepackage{amsfonts}
\usepackage{amsmath}
\usepackage{fullpage}
\usepackage{amssymb}
\usepackage{graphicx}
\usepackage{hyperref}
\usepackage{lmodern}
\usepackage{bbding}
\providecommand{\U}[1]{\protect\rule{.1in}{.1in}}
\newtheorem{theorem}{Theorem}

\newtheorem{conjecture}[theorem]{Conjecture}
\newtheorem{consequence}[theorem]{Consequence}
\newtheorem{corollary}[theorem]{Corollary}

\newtheorem{definition}[theorem]{Definition}

\newtheorem{lemma}[theorem]{Lemma}

\newtheorem{proposition}[theorem]{Proposition}
\newtheorem{remark}[theorem]{Remark}

\numberwithin{equation}{section}
\newenvironment{proof}[1][Proof]{\noindent\textbf{#1.} }{\ \rule{0.5em}{0.5em}}
\begin{document}

%%%%%%%%%%%%%%%%%%%%%%%%%%%%%%%%%%%%%%%%%%%%%%%%%%%%%%%%%%%%%%%%%%%%%%%%%%%%%%%%%

\title{\textbf{R\'enyi squashed entanglement, discord, and relative entropy
differences}}
\author{Kaushik P. Seshadreesan\thanks{Hearne Institute for Theoretical Physics and
Department of Physics and Astronomy, Louisiana State University, Baton Rouge,
Louisiana 70803, USA}
\and Mario Berta\thanks{Institute for Quantum Information and Matter, California
Institute of Technology, Pasadena, California 91125, USA}
\and Mark M. Wilde\footnotemark[1] \thanks{Center for Computation and Technology,
Louisiana State University, Baton Rouge, Louisiana 70803, USA} }
\maketitle

%%%%%%%%%%%%%%%%%%%%%%%%%%%%%%%%%%%%%%%%%%%%%%%%%%%%%%%%%%%%%%%%%%%%%%%%%%%%%%%%%

\begin{abstract}
The squashed entanglement quantifies the amount of entanglement in a bipartite quantum state, and it satisfies all of the axioms desired for an entanglement measure. The quantum discord is a measure of quantum
correlations that are different from those due to entanglement. What these
two measures have in common is that they are both based upon the conditional
quantum mutual information. In [Berta {\it et al.},
J.~Math.~Phys.~\textbf{56}, 022205 (2015)], we recently proposed R\'{e}nyi
generalizations of the conditional quantum mutual information of a tripartite state on $ABC$ (with $C$ being the conditioning system), which were shown to satisfy some properties that hold for the original quantity, such as non-negativity, duality, and monotonicity with respect to local
operations on the system $B$ (with it being left open to show that the
R\'{e}nyi quantity is monotone with respect to local operations on system $A$). Here we
define a R\'{e}nyi squashed entanglement and a R\'{e}nyi
quantum discord based on a R\'{e}nyi conditional quantum mutual information and investigate these quantities in detail. Taking
as a conjecture that the R\'{e}nyi conditional quantum mutual information is
monotone with respect to local operations on both systems $A$ and $B$, we prove that the
R\'{e}nyi squashed entanglement and the R\'{e}nyi quantum discord satisfy many of the properties of the respective original von Neumann entropy based quantities.
In our prior work~[Berta {\it et al.},
Phys.~Rev.~A~\textbf{91}, 022333 (2015)], we also detailed a procedure to obtain R\'{e}nyi generalizations of any
quantum information measure that is equal to a linear combination of von
Neumann entropies with coefficients chosen from the set $\{-1,0,1\}$. Here, we
extend this procedure to include differences of relative entropies. Using the
extended procedure and a conjectured monotonicity of the R\'enyi generalizations in the R\'enyi parameter, we discuss potential remainder terms for well known inequalities such as monotonicity of the relative entropy, joint convexity of the relative entropy, and the Holevo bound.
\end{abstract}

%%%%%%%%%%%%%%%%%%%%%%%%%%%%%%%%%%%%%%%%%%%%%%%%%%%%%%%%%%%%%%%%%%%%%%%%%%%%%%%%%

\section{Introduction}\label{sec:intro}

Quantum information theory is intimately connected to the
theory of quantum correlations. Information quantities such as the quantum
entropy, conditional entropy, mutual information, and conditional mutual
information underlie several measures of quantum correlations. For example,
both the squashed entanglement \cite{CW04} and the quantum discord
\cite{zurek} are based upon the conditional quantum mutual information. Also,
the operational meanings that the entropic quantities take in
information-theoretic protocols aid in understanding the correlation measures
defined based on them. The entropy finds operational meaning in the context of
quantum data compression \cite{S95}, the conditional entropy in quantum state
merging \cite{HOW05,HOW07}, the mutual information in erasure of total
correlations \cite{GPW05}, and the conditional mutual information in quantum
state redistribution \cite{DY08,YD09}.

The squashed entanglement of a bipartite state $\rho_{AB}$ is defined as
\cite{CW04}
\begin{equation}
E^{\text{sq}}(A;B)_{\rho}\equiv\frac{1}{2}\inf_{\omega_{ABE}}\left\{
I(A;B|E)_{\omega}:\text{Tr}_{E}\left\{  \omega_{ABE}\right\}  =\rho
_{AB}\right\}  ,\label{eq:sqent}
\end{equation}
where the infimum is with respect to all extensions $\omega_{ABE}$ of the
state $\rho_{AB}$. The conditional mutual information $I\left(  A;B|E\right)
_{\omega}$ in (\ref{eq:sqent}) is the linear combination
\begin{equation}
I\left(  A;B|E\right)  _{\omega}\equiv H\left(  AE\right)  _{\omega}+H\left(
BE\right)  _{\omega}-H\left(  E\right)  _{\omega}-H\left(  ABE\right)
_{\omega},\label{eq:cqmidef}
\end{equation}
with $H(F)_{\sigma}\equiv-$Tr$\{\sigma_{F}\log\sigma_{F}\}$ being the von
Neumann entropy of a state $\sigma_{F}$ on system $F$, where we unambiguously
let $\omega_{E}\equiv\ $Tr$_{AB}\{\omega_{ABE}\}$ denote the reduced density
operator on system $E$, for example. The squashed entanglement
satisfies all the desired properties of an entanglement measure.
It is convex, monotone with respect to local operations and
classical communication (LOCC), non-negative, superadditive in general, and
additive on tensor-product states \cite{CW04}. It is also asymptotically
continuous \cite{AF04}, and faithful (i.e., it takes the value zero if and only
if the state $\rho_{AB}$ is separable)~\cite{CW04,BCY11}. The squashed
entanglement finds application as an upper bound on distillable entanglement
and as a lower bound on entanglement cost \cite{CW04}. The squashed entanglement
of a quantum channel is defined as the maximum squashed entanglement that can
be registered by a sender and a receiver at the input and output of a channel
\cite{TGW13,TGW14}. This quantity is an upper bound on the quantum
communication capacity of any channel assisted by unlimited forward and
backward classical communication as well as the private capacity of a quantum
channel assisted by unlimited forward and backward public classical communication.

The quantum discord of a bipartite state $\rho_{AB}$ is defined as the gap
between the quantum mutual information $I\left(  A;B\right)  _{\rho}\equiv
H\left(  A\right)  _{\rho}+H\left(  B\right)  _{\rho}-H\left(  AB\right)
_{\rho}$ and the mutual information after one of the systems $A$ or $B$ has been
measured, where the latter is maximized over all measurements \cite{Z00,zurek}.
That is, if the measurement occurs on system $A$, then the quantum discord is defined
as%
\begin{equation}
D(\overline{A};B)_{\rho}\equiv I\left(  A;B\right)  _{\rho}-\max
_{\{\Lambda^{x}\}}I(X;B)_{\tau},
\end{equation}
with the overbar denoting the system being measured,
\begin{equation}
\tau_{XB}\equiv\sum_{x}|x\rangle\langle x|_{X}\otimes\operatorname{Tr}
_{A}\{(\Lambda_{A}^{x}\otimes I_{B})\rho_{AB}\},
\end{equation}
$\Lambda^{x}\geq0 \ \forall x$, $\sum_{x}\Lambda^{x}=I$,
and $\{\vert x\rangle\}$ is an orthonormal basis. (It is well known to be sufficient to optimize the quantum discord with respect to rank-one POVMs~\cite{KBCPV12}.) We think that it is an important conceptual realization that the quantum discord can be re-expressed as the following conditional quantum mutual information~\cite{P12}:
\begin{equation}
D(  \overline{A};B)  _{\rho}=\inf_{\left\{  \Lambda^{x}\right\}
}I\left(  E;B|X\right)  _{\tau},
\end{equation}
where the optimization is with respect to all possible POVMs acting on system $A$, with
$X$ being the classical output and $E$ being an environment for the
measurement map, so that
\begin{align}
\tau_{XEB}  &  \equiv U_{A\rightarrow XE}\rho_{AB}U_{A\rightarrow XE}^{\dag
},\\
U_{A\rightarrow XE}|\psi\rangle_{A}  &  \equiv\sum_{x}|x\rangle_{X}%
\otimes\left(  \sqrt{\Lambda^{x}}|\psi\rangle_{A}\otimes|x\rangle\right)
_{E}.
\end{align}
The quantum discord characterizes quantum correlations that are different from those due to entanglement. It is non-negative, invariant with respect to local unitary operations, and equal to zero if and only if the state is classical on the system being measured~\cite{KBCPV12}.

In quantum information theory, there has been interest recently in
generalizing information measures in terms of the R\'{e}nyi entropies. Such
R\'{e}nyi generalizations of information measures enable an improved
understanding of the original von Neumann entropy-based measures and also find applications
in scenarios beyond the traditional independent and identically (i.i.d.)
distributed resource setting \cite{P86,RennerThesis,D09,MH11,T12,MDSFT13,WWY13}. Since quantum information quantities underlie some quantum correlation measures, it is natural to think of R\'{e}nyi generalizations of those correlation measures as well.

In our prior work~\cite{BSW142}, we had proposed a procedure to obtain R\'{e}nyi
generalizations of quantum information measures that are equal to linear combinations
of von Neumann entropies with coefficients chosen from the set $\left\{
-1,0,1\right\}  $. In particular, R\'{e}nyi
generalizations of the conditional quantum mutual information of
(\ref{eq:cqmidef}) were proposed~\cite[Sections 5 and 6]{BSW14}. These
generalizations retain most of the desired properties of the von Neumann
entropy-based conditional quantum mutual information; e.g., they are non-negative, monotone
with respect to local quantum operations on one of the two systems $A$ or $B$, and obey
a duality property for four-party pure states $|\psi\rangle_{ABCD}
$~\cite[Section 7]{BSW14}. There it was left as an open question to determine
whether monotonicity also holds for local operations on the other system, with
numerical evidence supporting a positive answer.

In this work, based on the R\'{e}nyi
conditional quantum mutual informations, we define and explore the properties of
R\'{e}nyi generalizations of the squashed entanglement and the quantum discord.~\footnote{Note that what we define in this paper as `the' R{\'e}nyi squashed entanglement and `the' R{\'e}nyi quantum discord are only one of many possible definitions for the quantities based on the many different R{\'e}nyi conditional mutual informations defined in~\cite{BSW14}.}~\footnote{A different R\'enyi generalization of the quantum discord based on a difference of R{\'e}nyi mutual informations was proposed and
studied in \cite{MBPDS14}.} Additionally, we extend the scope of the procedure
for R\'enyi generalization of quantum information measures prescribed in
our former work to include quantum information measures which are
expressed as differences of relative entropies. The major potential
application of this latter contribution is in finding ``remainder
terms'' for monotonicity of relative entropy, joint convexity of relative
entropy, and the Holevo bound.

%%%%%%%%%%%%%%%%%%%%%%%%%%%%%%%%%%%%%%%%%%%%%%%%%%%%%%%%%%%%%%%%%%%%%%%%%%%%%%%%%

\section{Overview of results}

The main purpose of the present paper is to justify the R\'enyi squashed
entanglement and quantum discord quantities introduced here as
useful measures of quantum correlation. Complementary to this aim, we take the
results here as further justification that \cite{BSW14} has pinned down a
useful R\'enyi generalization of the conditional mutual information.
Furthermore, our work suggests that the axiomatic approach to entanglement
theory and discord does not uniquely single out the von Neumann entropy-based
squashed entanglement or discord (however, further effort is needed to justify
this claim in full). Finally, a potential application for future work is to
use the R\'enyi quantities in order to refine information-theoretic results
concerning the original von Neumann entropy-based quantities.

At the outset, it is important to note that this paper takes as a conjecture
that the R\'{e}nyi conditional quantum mutual information is monotone with respect to
local quantum operations on both systems $A$ and $B$ (see
Conjecture~\ref{conj:monoA}). As already stated in Section~\ref{sec:intro},
the aforementioned monotonicity was proven in \cite{BSW14} for local quantum
operations on one of the two systems $A$ or $B$, and it was left as an open
question to determine if the property holds for local quantum operations on the other system (with numerical evidence supporting a positive answer). Here, we assume the latter to hold. Some of the results presented in this work are based upon this conjecture.

Our contributions are as follows:

\begin{itemize}
\item We show that the R\'{e}nyi squashed entanglement defined here retains
many of the desired properties of the original von Neumann entropy based quantity, such as
monotonicity with respect to LOCC (up to the conjecture on the monotonicity of the
R\'{e}nyi conditional quantum mutual information with respect to local operations), convexity, and sub-additivity on tensor-product states. These properties hold only for particular regimes of
the R\'{e}nyi parameter $\alpha$. The first two properties together qualify
the R\'{e}nyi squashed entanglement as an entanglement monotone, again for
certain regimes of $\alpha$ and up to the conjecture.

\item The squashed entanglement is known to be
bounded from above by the
entanglement of formation \cite{CW04}. Here we define a R\'{e}nyi entanglement
of formation based on the R\'{e}nyi conditional entropy, and show that the
$\alpha$-R\'{e}nyi squashed entanglement is also upper bounded by a $\beta
$-R\'{e}nyi entanglement of formation, where $\beta=(2-\alpha)/\alpha$.

%Similarly, the squashed
%entanglement is known to be an upper bound on distillable entanglement of a
%bipartite state $\rho_{AB}$. We define a R\'enyi classical squashed
%entanglement, where the optimization is restricted to be over all classical
%extensions $\sigma_{ABE}=\sum p\left(  x\right)  \sigma_{AB}^{x}%
%\otimes|x\rangle\langle x|_{E}$, and we use it to show that the von Neumann
%classical squashed entanglement is a strong converse rate for entanglement
%distillation [THIS RESULT IS QUESTIONABLE NOW].

\item We show that the R\'{e}nyi discord is non-negative, invariant with respect to
local unitaries, and equal to zero if the state is a classical-quantum state.
Furthermore, we prove that it suffices to optimize the R\'{e}nyi discord
functional with respect to rank-one measurements. We then discuss an operational
characterization of bipartite quantum states with small von Neumann entropy-based quantum
discord, which would follow from \cite[Conjecture~34]{BSW14} (see just after \eqref{eq:duality} for the formal statement of this conjecture).

\item Finally, we extend the procedure to obtain R\'{e}nyi generalizations of
quantum information measures given in our prior work
\cite{BSW142} by showing how to obtain R\'{e}nyi generalizations of differences of
relative entropies. The main purpose of this last contribution is to study
potential remainder terms for monotonicity of the relative entropy, joint convexity of the relative entropy, and the Holevo bound.
\end{itemize}

%%%%%%%%%%%%%%%%%%%%%%%%%%%%%%%%%%%%%%%%%%%%%%%%%%%%%%%%%%%%%%%%%%%%%%%%%%%%%%%%%

\section{Preliminaries}\label{sec:notation}

\textbf{Norms, states, extensions, channels, and measurements.} Let $\mathcal{B}\left(  \mathcal{H}\right)  $ denote the
algebra of bounded linear operators acting on a Hilbert space $\mathcal{H}$.
We restrict ourselves to finite-dimensional Hilbert spaces throughout this
paper. For $\alpha\geq1$, we define the $\alpha$-norm of an operator $X \in \mathcal{B}(  \mathcal{H})$ as
\begin{equation}
\left\Vert X\right\Vert _{\alpha}\equiv
\left[\text{Tr}\{|X|^{\alpha
}\}\right]^{1/\alpha},\label{eq:a-norm}
\end{equation}
where $|X| \equiv \sqrt{X^{\dag}X}$ and we use the same notation even for the case $\alpha\in(0,1)$, when it is
not a norm. Let $\mathcal{B}\left(  \mathcal{H}\right)  _{+}$ denote the
subset of positive semi-definite operators, and let $\mathcal{B}\left(
\mathcal{H}\right)  _{++}$ denote the subset of positive definite operators.
We also write $X\geq0$ if $X\in\mathcal{B}\left(  \mathcal{H}\right)  _{+}$
and $Y>0$ if $Y\in\mathcal{B}\left(  \mathcal{H}\right)  _{++}$. An\ operator
$\rho$ is in the set $\mathcal{S}\left(  \mathcal{H}\right)  $\ of density
operators (or states) if $\rho\in\mathcal{B}\left(  \mathcal{H}\right)  _{+}$
and Tr$\left\{  \rho\right\}  =1$. An\ operator $\rho$ is in the set
$\mathcal{S}\left(  \mathcal{H}\right)  _{++}$ if it is a density operator and
positive definite. The tensor product of two Hilbert spaces $\mathcal{H}_{A}$
and $\mathcal{H}_{B}$ is denoted by $\mathcal{H}_{A}\otimes\mathcal{H}_{B}$ or
$\mathcal{H}_{AB}$.\ Given a multipartite density operator $\rho_{AB}
\in\mathcal{S}(\mathcal{H}_{A}\otimes\mathcal{H}_{B})$, we unambiguously write
$\rho_{A}=\ $Tr$_{B}\left\{  \rho_{AB}\right\}  $ for the reduced density
operator on system $A$. We use $\rho_{AB}$, $\sigma_{AB}$, $\tau_{AB}$,
$\omega_{AB}$, etc.~to denote general density operators in $\mathcal{S}
(\mathcal{H}_{A}\otimes\mathcal{H}_{B})$, while $\psi_{AB}$, $\varphi_{AB}$,
$\phi_{AB}$, etc.~denote rank-one density operators (pure states) in
$\mathcal{S}(\mathcal{H}_{A}\otimes\mathcal{H}_{B})$ (with it implicit, clear
from the context, and the above convention implying that $\psi_{A}$,
$\varphi_{A}$, $\phi_{A}$ are mixed if $\psi_{AB}$, $\varphi_{AB}$, $\phi
_{AB}$ are pure and entangled).

We also say that pure-state vectors $|\psi\rangle$ in $\mathcal{H}$ are
states. Any bipartite pure state $|\psi\rangle_{AB}$ in $\mathcal{H}_{AB}$ is written in Schmidt form as
\begin{equation}
\left\vert \psi\right\rangle _{AB}\equiv\sum_{i=0}^{d-1}\sqrt{\lambda_{i}
}\left\vert i\right\rangle _{A}\left\vert i\right\rangle _{B},
\end{equation}
where $\{|i\rangle_{A}\}$ and $\{|i\rangle_{B}\}$ form orthonormal bases in
$\mathcal{H}_{A}$ and $\mathcal{H}_{B}$, respectively, $d$ is the Schmidt rank
of the state, and $\sum_{i=0}^{d-1}\lambda_{i}=1$. By a maximally entangled
state, we mean a bipartite pure state of the form
\begin{equation}
\left\vert \Phi\right\rangle _{AB}\equiv\frac{1}{\sqrt{d}}\sum_{i=0}
^{d-1}\left\vert i\right\rangle _{A}\left\vert i\right\rangle _{B}.
\end{equation}

The trace distance between two quantum states $\rho,\sigma\in\mathcal{S}
\left(  \mathcal{H}\right)  $\ is equal to $\left\Vert \rho-\sigma\right\Vert
_{1}$. It has a direct operational interpretation in terms of the
distinguishability of these states. That is, if $\rho$ or $\sigma$ are
prepared with equal probability and the task is to distinguish them via some
quantum measurement, then the optimal success probability in doing so is equal
to $\left(  1+\left\Vert \rho-\sigma\right\Vert _{1}/2\right)  /2$.

Throughout this paper, we take the usual convention that $f\left(  A\right)
=\sum_{i:a_i > 0}f\left(  a_{i}\right)  \left\vert i\right\rangle \left\langle
i\right\vert $ when given a function $f$ and a positive semi-definite operator $A$ with
spectral decomposition $A=\sum_{i}a_{i}\left\vert i\right\rangle \left\langle
i\right\vert $. So this means that $A^{-1}$ is interpreted as a generalized
inverse, so that $A^{-1}=\sum_{i:a_i > 0}a_{i}^{-1}\left\vert i\right\rangle
\left\langle i\right\vert $, $\log\left(  A\right)  =\sum_{i:a_i > 0}\log\left(
a_{i}\right)  \left\vert i\right\rangle \left\langle i\right\vert $,
$\exp\left(  A\right)  =\sum_{i:a_i > 0}\exp\left(  a_{i}\right)  \left\vert
i\right\rangle \left\langle i\right\vert $, etc. The above convention for
$f\left(  A\right)  $ leads to the convention that $A^{0}$ denotes the
projection onto the support of $A$, i.e., $A^{0}=\sum_{i:a_i > 0}\left\vert
i\right\rangle \left\langle i\right\vert $. We employ the shorthand supp$(A)$
to refer to the support of an operator~$A$.

A linear map $\mathcal{N}_{A\rightarrow B}:\mathcal{B}\left(  \mathcal{H}
_{A}\right)  \rightarrow\mathcal{B}\left(  \mathcal{H}_{B}\right)  $\ is
positive if $\mathcal{N}_{A\rightarrow B}\left(  \sigma_{A}\right)
\in\mathcal{B}\left(  \mathcal{H}_{B}\right)  _{+}$ whenever $\sigma_{A}
\in\mathcal{B}\left(  \mathcal{H}_{A}\right)  _{+}$. Let id$_{A}$ denote the
identity map acting on a system $A$. A linear map $\mathcal{N}_{A\rightarrow
B}$ is completely positive if the map id$_{R}\otimes\mathcal{N}_{A\rightarrow
B}$ is positive for a reference system $R$ of arbitrary size. A linear map
$\mathcal{N}_{A\rightarrow B}$ is strictly positive if it takes positive
definite operators on the input to positive definite operators on the output,
and it is a strict completely positive map if id$_{R}\otimes\mathcal{N}
_{A\rightarrow B}$ is strictly positive for a reference system $R$ of
arbitrary size. A linear map $\mathcal{N}_{A\rightarrow B}$ is
trace-preserving if Tr$\left\{  \mathcal{N}_{A\rightarrow B}\left(  \tau
_{A}\right)  \right\}  =\ $Tr$\left\{  \tau_{A}\right\}  $ for all input
operators $\tau_{A}\in\mathcal{B}\left(  \mathcal{H}_{A}\right)  $. If a
linear map is completely positive and trace-preserving (CPTP), we say that it
is a quantum channel or quantum operation. A positive operator-valued measure
(POVM) is a set $\left\{  \Lambda^{m}\right\}  $ of positive semi-definite
operators such that $\sum_{m}\Lambda^{m}=I$.

An extension of a state $\rho_{A}\in\mathcal{S}\left(  \mathcal{H}_{A}\right)
$ is some state $\Omega_{RA}\in\mathcal{S}\left(  \mathcal{H}_{R}
\otimes\mathcal{H}_{A}\right)  $ such that $\mathrm{Tr}_{R}\left\{
\Omega_{RA}\right\}  =\rho_{A}$. An isometric extension $U_{A\rightarrow
BE}^{\mathcal{N}}$ of a channel $\mathcal{N}_{A\rightarrow B}$ acting on a
state $\rho_{A}\in\mathcal{S}(\mathcal{H}_{A})$ is a linear map that satisfies
the following:
\begin{equation}
\mathrm{Tr}_{E}\left\{  U_{A\rightarrow BE}^{\mathcal{N}}\rho_{A}
(U_{A\rightarrow BE}^{\mathcal{N}})^{\dag}\right\}  =\mathcal{N}_{A\rightarrow
B}\left(  \rho_{A}\right)  ,\ \ \ \ \ U_{\mathcal{N}}^{\dagger}U_{\mathcal{N}
}=I_{A},\ \ \ \ \ U_{\mathcal{N}}U_{\mathcal{N}}^{\dagger}=\Pi_{BE},
\end{equation}
where $\Pi_{BE}$ is a projection onto a subspace of the Hilbert space
$\mathcal{H}_{B}\otimes\mathcal{H}_{E}$.

\bigskip{}

\textbf{R\'{e}nyi entropies and information measures.}
The $\alpha$-R\'{e}nyi entropy and the
$\alpha$-R\'{e}nyi divergence, are defined respectively for a parameter $\alpha
\in\left(  0,1\right)  \cup\left(  1,\infty\right)  $ and probability
distributions $p$ and $q$ as%
\begin{align}
H_{\alpha}(p)  &  \equiv\frac{1}{1-\alpha}\log\sum_{x}\left[  p\left(
x\right)  \right]  ^{\alpha},\label{eq:Renyi-ent}\\
D_{\alpha}(p\Vert q)  &  \equiv\frac{1}{\alpha-1}\log\sum_{x}\left[  p\left(
x\right)  \right]  ^{\alpha}\left[  q\left(  x\right)  \right]  ^{1-\alpha},
\label{eq:cl-Renyi-rel-ent}%
\end{align}
where $\log$ denotes the natural logarithm here and throughout the paper. The
Shannon entropy and relative entropy are recovered in the limit as
$\alpha\rightarrow1$:%
\begin{align}
\lim_{\alpha\rightarrow1}H_{\alpha}(p)  &  =H(p)\equiv-\sum_{x}p\left(
x\right)  \log p\left(  x\right)  ,\\
\lim_{\alpha\rightarrow1}D_{\alpha}(p\Vert q)  &  =D(p\Vert q)\equiv\sum
_{x}p\left(  x\right)  \log\frac{p\left(  x\right)  }{q\left(  x\right)  }.
\end{align}
The quantum R\'{e}nyi entropy of a state $\rho$ is defined for $\alpha
\in\left(  0,1\right)  \cup\left(  1,\infty\right)  $ as%
\begin{equation}
H_{\alpha}\left(  \rho\right)  \equiv\frac{1}{1-\alpha}\log\operatorname{Tr}
\{ \rho^{\alpha}\} = \frac{\alpha}{1-\alpha}\log\left\Vert \rho\right\Vert
_{\alpha},
\end{equation}
and reduces to the von Neumann entropy in the limit as $\alpha\rightarrow1$:%
\begin{equation}
\lim_{\alpha\rightarrow1}H_{\alpha}\left(  \rho\right)  =H\left(  \rho\right)
.
\end{equation}

There are at least two ways to generalize the quantum relative entropy.
The R\'{e}nyi
relative entropy of order $\alpha\in\lbrack0,1)\cup(1,\infty)$ is defined as
\cite{P86}
\begin{equation}
D_{\alpha}(\rho\Vert\sigma)\equiv\left\{
\begin{array}
[c]{cc}
\frac{1}{\alpha-1}\log\text{Tr}\left\{  \left[  \text{Tr}\left\{
\rho\right\}  \right]  ^{-1}\rho^{\alpha}\sigma^{1-\alpha}\right\}   &
\text{if supp}\left(  \rho\right)  \subseteq\text{supp}\left(  \sigma\right)
\text{ or (}\alpha\in\lbrack0,1)\text{ and }\rho\not \perp \sigma\text{)}\\
+\infty & \text{otherwise}
\end{array}
\right.  ,\label{eq:Renyi-rel-ent}
\end{equation}
with the support conditions established in \cite{TCR09}. The R\'{e}nyi
relative entropy obeys monotonicity with respect to quantum operations for $\alpha
\in\lbrack0,1)\cup(1,2]$ \cite{P86}, in the sense that $D_{\alpha}(\rho
\Vert\sigma)\geq D_{\alpha}(\mathcal{N}\left(  \rho\right)  \Vert
\mathcal{N}\left(  \sigma\right)  )$ for a quantum operation $\mathcal{N}$.

The sandwiched R\'{e}nyi relative entropy \cite{MDSFT13,WWY13}\ is another
variant of the R\'{e}nyi relative entropy which has found a number of
applications recently in the context of strong converse theorems
\cite{WWY13,MO13,GW13,CMW14,TWW14}. It is defined for $\alpha\in
(0,1)\cup(1,\infty)$ as follows:
\begin{equation}
\widetilde{D}_{\alpha}\left(  \rho\Vert\sigma\right)  \equiv\left\{
\begin{array}
[c]{cc}
\frac{1}{\alpha-1}\log\left[  \left[  \text{Tr}\left\{  \rho\right\}  \right]
^{-1}\text{Tr}\left\{  \left(  \sigma^{\left(  1-\alpha\right)  /2\alpha}%
\rho\sigma^{\left(  1-\alpha\right)  /2\alpha}\right)  ^{\alpha}\right\}
\right]  &
\begin{array}
[c]{c}
\text{if supp}\left(  \rho\right)  \subseteq\text{supp}\left(  \sigma\right)
\text{ or}\\
\text{(}\alpha\in(0,1)\text{ and }\rho\not \perp \sigma\text{)}
\end{array}
\\
+\infty & \text{otherwise}
\end{array}
\right.  . \label{eq:def-sandwiched}
\end{equation}

The R\'{e}nyi conditional entropy for a bipartite state $\rho_{AB}$ is defined
for $\alpha\in(0,1)\cup(1,\infty)$ as
\begin{align}
H_{\alpha}\left(  A|B\right)   &  \equiv-\min_{\sigma_{B}}D_{\alpha}\left(
\rho_{AB}\Vert I_{A}\otimes\sigma_{B}\right)  \label{eq:condent}\\
&  =\frac{\alpha}{1-\alpha}\log\text{Tr}\left\{  \left(  \text{Tr}_{A}\left\{
\rho_{AB}^{\alpha}\right\}  \right)  ^{1/\alpha}\right\}  ,
\end{align}
where $D_{\alpha}$ is the R\'{e}nyi relative entropy defined in
(\ref{eq:Renyi-rel-ent}) and the second equality follows from a Sibson
identity~\cite{SW12}. The R\'{e}nyi quantum mutual information of a bipartite
state $\rho_{AB}$ is defined for $\alpha\in(0,1)\cup(1,\infty)$ as
\begin{align}
{I}_{\alpha}\left(  A;B\right)  _{\rho} &  \equiv\min_{\sigma_{B}}D_{\alpha
}\left(  \rho_{AB}\Vert\rho_{A}\otimes\sigma_{B}\right)  \label{RMI}\\
&  =\frac{\alpha}{\alpha-1}\log\operatorname{Tr}\left\{  \left(
\operatorname{Tr}_{A}\left\{  \rho_{A}^{1-\alpha}\rho_{AB}^{\alpha}\right\}
\right)  ^{\frac{1}{\alpha}\label{RMISibson}}\right\}  ,
\end{align}
where $D_{\alpha}$ is the R\'{e}nyi relative entropy defined in
(\ref{eq:Renyi-rel-ent}) and the second equality follows from a Sibson
identity~\cite[Corollary 8]{GW13}. Throughout this paper, we take our
definition of R\'{e}nyi conditional quantum mutual information for $\alpha
\in(0,1)\cup(1,\infty)$ to be as follows \cite[Definition~7]{BSW14}:
\begin{align}
I_{\alpha}\left(  A;B|E\right)  _{\rho} &  \equiv\inf_{\sigma_{BE}}D_{\alpha
}\left(  \rho_{ABE}\middle\Vert\left(  \rho_{AE}^{\left(  1-\alpha\right)
/2}\rho_{E}^{\left(  \alpha-1\right)  /2}\sigma_{BE}^{1-\alpha}\rho
_{E}^{\left(  \alpha-1\right)  /2}\rho_{AE}^{\left(  1-\alpha\right)
/2}\right)^{\frac{1}{1-\alpha}}  \right)  \label{eq:optrenyicmi}\\
&  =\frac{\alpha}{\alpha-1}\log\text{Tr}\left\{  \left(  \rho_{E}^{\left(
\alpha-1\right)  /2}\text{Tr}_{A}\left\{  \rho_{AE}^{\left(  1-\alpha\right)
/2}\rho_{ABE}^{\alpha}\rho_{AE}^{\left(  1-\alpha\right)  /2}\right\}
\rho_{E}^{\left(  \alpha-1\right)  /2}\right)  ^{1/\alpha}\right\}
,\label{eq:CMI-sibson}
\end{align}
where $D_{\alpha}$ is the R\'{e}nyi relative entropy defined in
(\ref{eq:Renyi-rel-ent}) and the second equality follows from a Sibson
identity~\cite[Proposition 8]{BSW14}. The R\'{e}nyi conditional quantum mutual
information is non-negative and obeys monotonicity with respect to local CPTP maps on
system $B$, i.e., for $\alpha\in(0,1)\cup(1,2]$,
\begin{equation}
I_{\alpha}\left(  A;B|E\right)  _{\rho}\geq I_{\alpha}\left(  A;B^{\prime
}|E\right)  _{\omega},
\end{equation}
where $\omega_{AB^{\prime}E}\equiv\mathcal{N}_{B\rightarrow B^{\prime}}
(\rho_{ABE})$ and $\mathcal{N}_{B\rightarrow B^{\prime}}$ is any CPTP map
acting on system $B$ alone. For a four-party pure state $\phi_{ABCD}$, the
R\'{e}nyi conditional quantum mutual information satisfies a duality property
\cite[Theorem 32]{BSW14} given by
\begin{equation}
I_{\alpha}\left(  A;B|C\right)  _{\phi}=I_{\alpha}\left(  B;A|D\right)
_{\phi}.\label{eq:duality}
\end{equation}
It is conjectured \cite[Conjecture~34]{BSW14} that the R\'enyi conditional mutual information $I_\alpha(A;B|E)$ is monotone non-decreasing with respect to the R\'enyi parameter $\alpha$.

Some results in this paper rely on the following conjecture:

\begin{conjecture}\label{conj:monoA}
The R\'enyi conditional quantum mutual information obeys monotonicity
with respect to CPTP maps on system $A$, i.e., for $\rho_{ABE}\in\mathcal{S}(\mathcal{H}_{ABE})$ and $\alpha\in(0,1)\cup(1,2]$,
\begin{equation}
I_{\alpha}\left(  A;B|E\right)  _{\rho}\geq I_{\alpha}\left(  A^{\prime
};B|E\right)  _{\tau},
\end{equation}
where $\tau_{A^{\prime}BE} \equiv\mathcal{M}_{A\rightarrow A^{\prime}}
(\rho_{ABE})$ and $\mathcal{M}_{A\rightarrow A^{\prime}}$ is a CPTP map acting
on system $A$ alone.
\end{conjecture}

\section{R\'enyi Squashed Entanglement}

In this section, we explore a R\'{e}nyi squashed entanglement,
assuming Conjecture~\ref{conj:monoA}. We show various properties
of the quantity which justify it is a valid entanglement measure (again
assuming Conjecture~\ref{conj:monoA}). We also establish a relation between the R\'{e}nyi squashed entanglement and the R\'{e}nyi entanglement of formation. We end the section with Table~\ref{sqentprops}, which summarizes the properties of the R\'enyi squashed entanglement and draws a comparison with those of the von Neumann entropy based squashed entanglement.

\begin{definition}
The R\'enyi squashed entanglement of $\rho_{AB}\in\mathcal{S}(\mathcal{H}_{AB})$ is defined
for $\alpha\in(0,1)\cup(1,\infty)$ as\label{def:RSE}
\begin{equation}
E_{\alpha}^{\operatorname{sq}}\left(  A;B\right)  _{\rho}\equiv\frac{1}{2}
\inf_{\omega_{ABE}}\left\{  I_{\alpha}\left(  A;B|E\right)  _{\omega}
:\rho_{AB}=\operatorname{Tr}_{E}\left\{  \omega_{ABE}\right\}  \right\}  ,
\label{eq:renyisqent}
\end{equation}
where the infimum is with respect to all extensions $\omega_{ABE}$ of the
state $\rho_{AB}$ and $I_{\alpha}\left(  A;B|E\right)  _{\omega}$ is the
R\'enyi conditional quantum mutual information from \eqref{eq:CMI-sibson}.
\end{definition}

%%%%%%%%%%%%%%%%%%%%%%%%%%%%%%%%%%%%%%%%%%%%%%%%%%%%%%%%%%%%%%%%%%%%%%%%%%%%%%%%%

\subsection{Entanglement monotone}

Assuming Conjecture \ref{conj:monoA}, we now show that the R\'{e}nyi squashed
entanglement of Definition \ref{def:RSE} is an entanglement monotone for $\alpha\in\left(  0,1\right)$ (see~\cite{PV07, HHHH09} for a review of the
axiomatic approach to entanglement measures). We say that a quantity is an entanglement monotone if it is an LOCC monotone and convex~\cite{Vida00}. We prove that the R\'{e}nyi squashed entanglement is an LOCC monotone for $\alpha
\in\left(  0,1\right)  \cup(1,2]$ (assuming Conjecture~\ref{conj:monoA}) and
convex for $\alpha\in\left(  0,1\right)$. Before beginning,
we recall that local operations and classical communication
consists of compositions of the following operations:

\begin{enumerate}
\item Alice performs a quantum instrument, which has both a quantum and
classical output. She forwards the classical output to Bob, who then performs
a quantum operation conditioned on the classical data received.

\item The situation is reversed, with Bob performing the initial instrument,
who forwards the classical data to Alice, who then performs a quantum
operation conditioned on the classical data.
\end{enumerate}
Thus, in order to establish a measure as an LOCC\ monotone, it suffices to
show that it is non-increasing with respect to local operations and invariant with respect to the
sharing of classical data.

\begin{consequence}
Assuming Conjecture \ref{conj:monoA}, the R\'{e}nyi squashed entanglement
$E_{\alpha}^{\operatorname{sq}}\left(  A;B\right)  _{\rho}$ is monotone with respect to
local operations for $\alpha\in(0,1)\cup(1,2]$:\label{lem:localmono}
\begin{equation}
E_{\alpha}^{\operatorname{sq}}\left(  A;B\right)  _{\rho}\geq E_{\alpha
}^{\operatorname{sq}}(A^{\prime};B^{\prime})_{\sigma},
\end{equation}
where $\rho_{AB}\in\mathcal{S}(\mathcal{H}_{AB})$, $\sigma_{A^{\prime}B^{\prime}}\equiv\left(  \mathcal{N}_{A\rightarrow
A^{\prime}}\otimes\mathcal{M}_{B\rightarrow B^{\prime}}\right)  \left(
\rho_{AB}\right)  $, and $\mathcal{N}_{A\rightarrow A^{\prime}}$, $\mathcal{M}_{B\rightarrow B^{\prime}}$ are CPTP maps.
\end{consequence}

\begin{proof}
This follows directly from monotonicity with respect to local operations of the
R\'{e}nyi conditional mutual information $I_{\alpha}\left(  A;B|E\right)
_{\rho}$.
\end{proof}

\begin{consequence}
Assuming Conjecture~\ref{conj:monoA}, the R\'{e}nyi squashed entanglement
$E_{\alpha}^{\operatorname{sq}}\left(  A;B\right)  _{\rho}$ is invariant with respect to
classical communication for $\alpha\in(0,1)\cup(1,2]$, in the sense
that\label{lem:ccmono}
\begin{equation}
E_{\alpha}^{\operatorname{sq}}\left(  AX_{A};B\right)  _{\rho}=E_{\alpha
}^{\operatorname{sq}}\left(  AX_{A};BX_{B}\right)  _{\rho}=E_{\alpha
}^{\operatorname{sq}}\left(  A;BX_{B}\right)  _{\rho},
\end{equation}
where
\begin{equation}
\rho_{X_{A}X_{B}AB}\equiv\sum_{x}p_{X}\left(  x\right)  \left\vert
x\right\rangle \left\langle x\right\vert _{X_{A}}\otimes\left\vert
x\right\rangle \left\langle x\right\vert _{X_{B}}\otimes\rho_{AB}^{x}.
\end{equation}
\end{consequence}

\begin{proof}
Let $\rho_{X_{A}ABE}$ be any extension of $\rho_{X_{A}AB}$, taking the form
\begin{equation}
\sum_{x}p_{X}\left(  x\right)  \left\vert x\right\rangle \left\langle
x\right\vert _{X_{A}}\otimes\rho_{ABE}^{x},
\end{equation}
where $\rho_{ABE}^{x}$ is an extension of $\rho_{AB}^{x}$. Let $\left\vert
\varphi^{\rho}\right\rangle _{X_{A}X_{R}ABER}$ be a purification of
$\rho_{X_{A}ABE}$, taken as
\begin{equation}
\left\vert \varphi^{\rho}\right\rangle _{X_{A}X_{R}ABER}\equiv\sum_{x}
\sqrt{p_{X}\left(  x\right)  }\left\vert x\right\rangle _{X_{A}}\left\vert
x\right\rangle _{X_{R}}\left\vert \varphi^{\rho_{x}}\right\rangle _{ABER},
\end{equation}
where $\left\vert \varphi^{\rho_{x}}\right\rangle _{ABER}$ purifies
$\rho_{ABE}^{x}$. Let $\overline{\varphi}_{X_{A}X_{B}X_{R}ABER}^{\rho}$ be
defined as
\begin{equation}
\overline{\varphi}_{X_{A}X_{B}X_{R}ABER}^{\rho}\equiv\sum_{x}p_{X}\left(
x\right)  \left\vert x\right\rangle \left\langle x\right\vert _{X_{A}}
\otimes\left\vert x\right\rangle \left\langle x\right\vert _{X_{B}}
\otimes\left\vert x\right\rangle \left\langle x\right\vert _{X_{R}}
\otimes\varphi_{ABER}^{\rho_{x}}.
\end{equation}
Let $\theta_{X_{A}X_{B}X_{R}ABERF}$ be a purification of $\overline{\varphi
}_{X_{A}X_{B}X_{R}ABER}^{\rho}$. We then have that
\begin{align}
I_{\alpha}\left(  AX_{A};B|E\right)  _{\rho} &  =I_{\alpha}\left(
B;AX_{A}|X_{R}R\right)  _{\varphi^{\rho}}\\
&  \geq I_{\alpha}\left(  B;AX_{A}|X_{R}R\right)  _{\overline{\varphi}^{\rho}
}\\
&  =\frac{\alpha}{\alpha-1}\log\sum_{x}p_{X}\left(  x\right)  \exp\left\{
\left(  \frac{\alpha-1}{\alpha}\right)  I_{\alpha}\left( B ;AX_{A}|R\right)
_{\left\vert x\right\rangle \left\langle x\right\vert \otimes\varphi^{\rho
_{x}}}\right\}  \\
&  =\frac{\alpha}{\alpha-1}\log\sum_{x}p_{X}\left(  x\right)  \exp\left\{
\left(  \frac{\alpha-1}{\alpha}\right)  I_{\alpha}\left(  
BX_{B};AX_{A}|R\right)  _{\left\vert x\right\rangle \left\langle x\right\vert
\otimes\left\vert x\right\rangle \left\langle x\right\vert \otimes
\varphi^{\rho_{x}}}\right\}  \\
&  =I_{\alpha}\left(  BX_{B};AX_{A}|X_{R}R\right)  _{\overline{\varphi}^{\rho
}}\\
&  =I_{\alpha}\left(  AX_{A};BX_{B}|EF\right)  _{\theta}\\
&  \geq2E_{\alpha}^{\operatorname{sq}}\left(  AX_{A};BX_{B}\right)  _{\rho}.
\end{align}
The first equality follows from duality of the R\'{e}nyi conditional quantum
mutual information given in (\ref{eq:duality}). The first inequality is a
result of monotonicity with respect to a local dephasing operation $\sum_{x}\left\vert
x\right\rangle \left\langle x\right\vert _{X_{A}}\left(  \cdot\right)
\left\vert x\right\rangle \left\langle x\right\vert _{X_{A}}$. The second equality follows from
Lemma~\ref{lem:classical-cond} given in Appendix \ref{sub:CononC}. The third
equality is an application of Lemma~\ref{lem:tensor-prod} given in Appendix
\ref{sub:Invaritens}. The fourth equality is from another application of
Lemma~\ref{lem:classical-cond}. The fifth equality is from another application
of duality given in (\ref{eq:duality}). The final inequality is a result of
Definition \ref{def:RSE}. Since the inequality holds for any extension of
$\rho_{X_{A}AB}$, it follows that
\begin{equation}
E_{\alpha}^{\operatorname{sq}}\left(  AX_{A};B\right)  _{\rho}\geq E_{\alpha
}^{\operatorname{sq}}\left(  AX_{A};BX_{B}\right)  _{\rho}.\label{eq:ineq1}
\end{equation}
By a similar line of reasoning (which relies on
Conjecture \ref{conj:monoA}), it follows that
\begin{equation}
E_{\alpha}^{\operatorname{sq}}\left(  A;BX_{B}\right)  _{\rho}\geq E_{\alpha
}^{\operatorname{sq}}\left(  AX_{A};BX_{B}\right)  _{\rho}.
\end{equation}
Finally, from monotonicity of $E_{\alpha}^{\operatorname{sq}}$ with respect to local
operations (which relies on
Conjecture \ref{conj:monoA}), we can conclude that
\begin{align}
E_{\alpha}^{\operatorname{sq}}\left(  AX_{A};BX_{B}\right)  _{\rho} &  \geq
E_{\alpha}^{\operatorname{sq}}\left(  AX_{A};B\right)  _{\rho},\\
E_{\alpha}^{\operatorname{sq}}\left(  AX_{A};BX_{B}\right)  _{\rho} &  \geq
E_{\alpha}^{\operatorname{sq}}\left(  A;BX_{B}\right)  _{\rho}
.\label{eq:ineq4}
\end{align}
So (\ref{eq:ineq1})-(\ref{eq:ineq4}) give the statement of the lemma.
\end{proof}

\begin{consequence}
Assuming Conjecture \ref{conj:monoA}, the R\'enyi squashed entanglement
$E_{\alpha}^{\operatorname{sq}}\left(  A;B\right)  _{\rho}$ is an LOCC
monotone for $\alpha\in(0,1)\cup(1,2]$. \label{thm:loccmono}
\end{consequence}

\begin{proof}
This consequence follows from Consequences \ref{lem:localmono} and \ref{lem:ccmono}.
\end{proof}

\begin{proposition}\label{prop:convexity}
The R\'{e}nyi squashed entanglement $E_{\alpha
}^{\operatorname{sq}}\left(  A;B\right)  _{\rho}$ is convex for $\alpha
\in(0,1)$, in the sense that
\begin{equation}
\sum_{x}p_{X}\left(  x\right)  E_{\alpha}^{\operatorname{sq}}\left(
A;B\right)  _{\rho^{x}}\geq E_{\alpha}^{\operatorname{sq}}\left(  A;B\right)
_{\overline{\rho}},
\end{equation}
where
\begin{equation}
\overline{\rho}_{AB}\equiv\sum_{x}p_{X}\left(  x\right)  \rho_{AB}^{x}.
\end{equation}
\end{proposition}

\begin{proof}
Let $\rho_{ABE}^{x}$ be any extension of $\rho_{AB}^{x}$. Let
\begin{equation}
\rho_{XABE}\equiv\sum_{x}p_{X}\left(  x\right)  \left\vert x\right\rangle
\left\langle x\right\vert _{X}\otimes\rho_{ABE}^{x}.
\end{equation}
Observe that $\rho_{XABE}$ is an extension of $\overline{\rho}_{AB}$, where
the extension systems are $X$ and $E$. So then
\begin{align}
2E_{\alpha}^{\operatorname{sq}}\left(  A;B\right)  _{\overline{\rho}}  &  \leq
I_{\alpha}\left(  A;B|EX\right) \\
&  =\frac{\alpha}{\alpha-1}\log\sum_{x}p_{X}\left(  x\right)  \exp\left\{
\left(  \frac{\alpha-1}{\alpha}\right)  I_{\alpha}\left(  A;B|E\right)
_{\rho^{x}}\right\} \\
&  \leq\sum_{x}p_{X}\left(  x\right)  \left[  \frac{\alpha}{\alpha-1}\log
\exp\left\{  \left(  \frac{\alpha-1}{\alpha}\right)  I_{\alpha}\left(
A;B|E\right)  _{\rho^{x}}\right\}  \right] \\
&  =\sum_{x}p_{X}\left(  x\right)  I_{\alpha}\left(  A;B|E\right)  _{\rho^{x}
}.
\end{align}
The first inequality follows from the definition of $E_{\alpha}
^{\operatorname{sq}}\left(  A;B\right)  _{\overline{\rho}}$, and the second
inequality follows from convexity of $-\log$. Since the inequality holds for
any extension of each $\rho_{ABE}^{x}$, we can conclude the statement of the proposition.
\end{proof}

\begin{consequence}
Assuming Conjecture \ref{conj:monoA}, the R\'enyi squashed entanglement
$E_{\alpha}^{\operatorname{sq}}\left(  A;B\right)  _{\rho}$ is an entanglement
monotone for $\alpha\in\left(  0,1\right)  $.
\end{consequence}

\begin{proof}
This consequence follows from Consequence \ref{thm:loccmono} and Proposition
\ref{prop:convexity}.
\end{proof}

%%%%%%%%%%%%%%%%%%%%%%%%%%%%%%%%%%%%%%%%%%%%%%%%%%%%%%%%%%%%%%%%%%%%%%%%%%%%%%%%%

\subsection{Other properties}

Having shown that the R\'enyi squashed entanglement is an entanglement
monotone (up to Conjecture~\ref{conj:monoA}), we now prove more properties of the quantity that bolster it as a valid entanglement measure.

\begin{proposition}\label{faithprop}
For separable states $\rho_{AB}\in\mathcal{S}(\mathcal{H}_{AB})$ the R\'enyi squashed entanglement $E_{\alpha}^{\operatorname{sq}}\left(  A;B\right)  _{\rho}$ vanishes for $\alpha\in(0,1)\cup(1,\infty)$.
\end{proposition}

\begin{proof}
We can take a convex decomposition of any separable state as follows:
\begin{equation}
\rho_{AB}=\sum_{x}p_{X}\left(  x\right)  \left\vert \psi_{x}\right\rangle
\left\langle \psi_{x}\right\vert _{A}\otimes\left\vert \phi_{x}\right\rangle
\left\langle \phi_{x}\right\vert _{B} .
\end{equation}
Then an extension is
\begin{equation}
\rho_{ABX}=\sum_{x}p_{X}\left(  x\right)  \left\vert \psi_{x}\right\rangle
\left\langle \psi_{x}\right\vert _{A}\otimes\left\vert \phi_{x}\right\rangle
\left\langle \phi_{x}\right\vert _{B}\otimes\left\vert x\right\rangle
\left\langle x\right\vert _{X}.
\end{equation}
In this case, we apply Lemma~\ref{lem:classical-cond}\ to find that
\begin{equation}
I_{\alpha}\left(  A;B|X\right)  =\frac{\alpha}{\alpha-1}\log\sum_{x}
p_{X}(  x)  \exp\left\{\left(  \frac{\alpha-1}{\alpha}\right)  I_{\alpha
}\left(  A;B\right)  _{\psi_{x}\otimes\phi_{x}} \right\} .
\end{equation}
The R\'{e}nyi mutual information is equal to zero for any product state. This
is easily seen from (\ref{RMI}) and the fact that the R\'{e}nyi relative
entropy is non-negative. This implies that
\begin{align}
I_{\alpha}\left(  A;B|X\right)=\frac{\alpha}{\alpha-1}\log\sum_{x}
p_{X}\left(  x\right)=0.
\end{align}

\end{proof}
\begin{remark}
Assuming~\cite[Conjecture 34]{BSW14}, it follows from the faithfulness of the original von Neumann entropy-based squashed entanglement \cite{BCY11} and Proposition~\ref{faithprop} that the R\'enyi squashed entanglement $E_{\alpha}^{\operatorname{sq}}\left(  A;B\right)  _{\rho}$ is faithful for $\alpha\in(1,\infty)$ (i.e., equals zero iff the state $\rho_{AB}$ is separable).
\end{remark}

\begin{proposition}
On tensor-product states $\rho_{A_{1}A_{2}B_{1}B_{2}}\equiv\sigma_{A_{1}B_{1}}\otimes\tau_{A_{2}B_{2}}\in\mathcal{S}(\mathcal{H}_{A_{1}A_{2}B_{1}B_{2}})$ the R\'enyi squashed entanglement is subadditive for $\alpha\in(0,1)\cup(1,\infty)$, in the sense that
\begin{equation}
E_{\alpha}^{\operatorname{sq}}\left(  A_{1}A_{2};B_{1}B_{2}\right)  _{\rho
}\leq E_{\alpha}^{\operatorname{sq}}\left(  A_{1};B_{1}\right)  _{\sigma}+
E_{\alpha}^{\operatorname{sq}}\left(  A_{2};B_{2}\right)  _{\tau}.
\end{equation}
\end{proposition}

\begin{proof}
Let $\sigma_{A_{1}B_{1}E_{1}}$ be some extension of $\sigma_{A_{1}B_{1}}$, and
let $\tau_{A_{2}B_{2}E_{2}}$ be some extension of $\tau_{A_{2}B_{2}}$. Then
\begin{equation}
\omega_{A_{1}B_{1}E_{1}A_{2}B_{2}E_{2}}\equiv\sigma_{A_{1}B_{1}E_{1}}
\otimes\tau_{A_{2}B_{2}E_{2}}
\end{equation}
is an extension of $\rho_{A_{1}A_{2}B_{1}B_{2}}$, so that
\begin{align}
2E_{\alpha}^{\operatorname{sq}}\left(  A_{1}A_{2};B_{1}B_{2}\right)  _{\rho}
&  \leq I_{\alpha}\left(  A_{1}A_{2};B_{1}B_{2}|E_{1}E_{2}\right)  _{\omega}\\
&  =I_{\alpha}\left(  A_{1};B_{1}|E_{1}\right)  _{\sigma}+I_{\alpha}\left(
A_{2};B_{2}|E_{2}\right)  _{\tau}.
\end{align}
The inequality is from the definition of R\'enyi squashed entanglement, and
the equality follows by direct substitution into (\ref{eq:CMI-sibson}). Since
the inequality is independent of the particular extensions of $\sigma
_{A_{1}B_{1}}$ and $\tau_{A_{2}B_{2}}$, the statement of the proposition follows.
\end{proof}

%%%%%%%%%%%%%%%%%%%%%%%%%%%%%%%%%%%%%%%%%%%%%%%%%%%%%%%%%%%%%%%%%%%%%%%%%%%%%%%%%

\subsection{Relations to R\'enyi entropy of entanglement and R\'enyi entanglement of formation}

The entropy of entanglement~\cite{BBPS96} and the entanglement of
formation~\cite{BDSW96} are among the earliest proposed measures of
entanglement for pure and mixed bipartite states, respectively. While the
former is defined as the von Neumann entropy of either reduced density
operators of a bipartite pure state, the latter is based upon the entropy of
entanglement via an extended convex roof construction. The entanglement of
formation of a bipartite state $\rho_{AB}$ is defined as
\begin{equation}
E^{F}(A;B)_{\rho}\equiv\min_{\left\{  p_{X}\left(  x\right)  ,\left\vert
\psi^{x}\right\rangle _{AB}\right\}  }\left\{  \sum_{x}p_{X}(x)E\left(
\psi_{AB}^{x}\right)  :\rho_{AB}=\sum_{x}p_{X}\left(  x\right)  \left\vert
\psi^{x}\right\rangle \left\langle \psi^{x}\right\vert _{AB}\right\}
,\label{eq:EF}
\end{equation}
where $E\left(  \psi_{AB}^{x}\right)  =H(A)_{\psi^{x}}$ is the entropy of
entanglement of pure state $|\psi^{x}\rangle_{AB}$. Note that the expression
in (\ref{eq:EF}) can be written more concisely in terms of the conditional
entropy $H\left(  A|X\right)  _{\sigma}$, where $\sigma_{XAB}$ is a
classical-quantum state of the following form:
\begin{equation}
\sigma_{XAB}\equiv\sum_{x}p_{X}\left(  x\right)  \left\vert x\right\rangle
\left\langle x\right\vert _{X}\otimes\left\vert \psi^{x}\right\rangle
\left\langle \psi^{x}\right\vert _{AB}.\label{cqstate1}
\end{equation}
That is,
\begin{equation}
E^{F}(A;B)_{\rho}=\min_{\left\{  p_{X}\left(  x\right)  ,\left\vert \psi
^{x}\right\rangle _{AB}\right\}  }\left\{  H(A|X)_{\sigma}:\rho_{AB}=\sum
_{x}p_{X}\left(  x\right)  \left\vert \psi^{x}\right\rangle \left\langle
\psi^{x}\right\vert _{AB}\right\}.
\end{equation}
We now consider a R\'{e}nyi entropy of entanglement and a R\'{e}nyi
entanglement of formation defined based upon the R\'{e}nyi entropy and the
R\'{e}nyi conditional entropy of (\ref{eq:condent}), respectively. We also
derive some relationships between the R\'{e}nyi squashed entanglement and
these quantities. While the R\'{e}nyi entropy of entanglement of a pure state
$\psi_{AB}^{x}$ is simply the R\'{e}nyi entropy $H_{\alpha}(A)_{\psi}$, the
R\'{e}nyi entanglement of formation of a state $\rho_{AB}$ is defined as follows.

\begin{definition}
The R\'enyi entanglement of formation of $\rho_{AB}\in\mathcal{S}(\mathcal{H}_{AB})$ is defined for $\alpha\in\left(  0,1\right)  \cup\left(  1,\infty\right)  $ as
\begin{equation}
E_{\alpha}^{F}\left(  A;B\right)  _{\rho}\equiv\inf_{\left\{  p_{X}\left(
x\right)  ,\left\vert \psi^{x}\right\rangle _{AB}\right\}  }\left\{
H_{\alpha}\left(  A|X\right)  _{\sigma}:\rho_{AB}=\sum_{x}p_{X}\left(
x\right)  \left\vert \psi^{x}\right\rangle \left\langle \psi^{x}\right\vert
_{AB}\right\}  , \label{eq:rentform}
\end{equation}
where $\sigma_{XAB}$ is a classical-quantum state of the form as in \eqref{cqstate1}.
\end{definition}

\begin{proposition}
The R\'{e}nyi entanglement of formation of $\rho_{AB}\in\mathcal{S}(\mathcal{H}_{AB})$ can be written as
\begin{equation}
E_{\alpha}^{F}\left(  A;B\right)  _{\rho}=\min_{\left\{  p_{X}\left(
x\right)  ,\left\vert \psi^{x}\right\rangle _{AB}\right\}  }\left\{
\frac{\alpha}{1-\alpha}\log\sum_{x}p_{X}\left(  x\right)  \left[
\operatorname{Tr}\left\{  \left(  \psi_{A}^{x}\right)  ^{\alpha}\right\}
^{1/\alpha}\right]  :\rho_{AB}=\sum_{x}p_{X}\left(  x\right)  \left\vert
\psi^{x}\right\rangle \left\langle \psi^{x}\right\vert _{AB}\right\}
.\label{eq:rentformclosed}
\end{equation}
\end{proposition}

\begin{proof}
By applying the Sibson identity mentioned in~\cite{SW12} to $H_{\alpha}\left(
A|X\right)  _{\sigma}$ in (\ref{eq:rentform}), we find that
\begin{align}
H_{\alpha}\left(  A|X\right)  _{\sigma} &  =\frac{\alpha}{1-\alpha}
\log\text{Tr}\left\{  \left(  \text{Tr}_{A}\left\{  \sigma_{AX}^{\alpha
}\right\}  \right)  ^{1/\alpha}\right\}  \\
&  =\frac{\alpha}{1-\alpha}\log\text{Tr}\left\{  \left(  \text{Tr}_{A}\left\{
\sum_{x}p_{X}^{\alpha}\left(  x\right)  \left\vert x\right\rangle \left\langle
x\right\vert _{X}\otimes\left(  \psi_{A}^{x}\right)  ^{\alpha}\right\}
\right)  ^{1/\alpha}\right\}  \\
&  =\frac{\alpha}{1-\alpha}\log\text{Tr}\left\{  \left(  \left\{  \sum
_{x}p_{X}^{\alpha}\left(  x\right)  \left\vert x\right\rangle \left\langle
x\right\vert _{X}\text{Tr}\left\{  \left(  \psi_{A}^{x}\right)  ^{\alpha
}\right\}  \right\}  \right)  ^{1/\alpha}\right\}  \\
&  =\frac{\alpha}{1-\alpha}\log\sum_{x}p_{X}\left(  x\right)  \left[
\text{Tr}\left\{  \left(  \psi_{A}^{x}\right)  ^{\alpha}\right\}  ^{1/\alpha
}\right]  .
\end{align}
Finally, by an application of the Carath\'{e}odory theorem, we know that the
infimum can be replaced by a minimum because $2^{\frac{1-\alpha}{\alpha
}H_{\alpha}\left(  A|X\right)  _{\sigma}}$ can be written as a convex
combination of Tr$\left\{  \left(  \psi_{A}^{x}\right)  ^{\alpha}\right\}
^{1/\alpha}$, which is a continuous function of the marginals of the elements
of a convex decomposition of the state $\rho_{AB}$ and these marginals are
elements of a compact set.
\end{proof}

\begin{proposition}\label{prop:squashed-pure}
For a pure state $\phi_{AB}\in\mathcal{S}(\mathcal{H}_{AB})$, the R\'enyi squashed entanglement is related to the R\'enyi entropy of entanglement as
\begin{equation}
E_{\alpha}^{\operatorname{sq}}\left(  A;B\right)  _{\phi}=H_{\left(
2-\alpha\right)  /\alpha}\left(  A\right)  _{\phi},
\end{equation}
for $\alpha\in(0,1)\cup(1,2]$.
\end{proposition}

\begin{proof}
Consider that any extension of a pure state $\phi_{AB}$ is of the form
$\phi_{AB}\otimes\omega_{E}$. So applying Lemma~\ref{lem:tensor-prod}, we find
that
\begin{equation}
I_{\alpha}\left(  A;B|E\right)  _{\phi\otimes\omega}=I_{\alpha}\left(
A;B\right)  _{\phi} .
\end{equation}
So now it is a matter of evaluating the R\'{e}nyi mutual information of a pure
state. Consider that
\begin{align}
I_{\alpha}\left(  A;B\right)  _{\phi} &  =\frac{\alpha}{\alpha-1}\log
\text{Tr}\left\{  \left(  \text{Tr}_{A}\left\{  \phi_{A}^{1-\alpha}\phi
_{AB}^{\alpha}\right\}  \right)  ^{1/\alpha}\right\}  \\
&  =\frac{\alpha}{\alpha-1}\log\text{Tr}\left\{  \left(  \text{Tr}_{A}\left\{
\phi_{A}^{1-\alpha}\phi_{AB}\right\}  \right)  ^{1/\alpha}\right\}  \\
&  =\frac{\alpha}{\alpha-1}\log\text{Tr}\left\{  \left(  \text{Tr}_{A}\left\{
\phi_{B}^{1-\alpha}\phi_{AB}\right\}  \right)  ^{1/\alpha}\right\}  \\
&  =\frac{\alpha}{\alpha-1}\log\text{Tr}\left\{  \left(  \phi_{B}^{1-\alpha
}\phi_{B}\right)  ^{1/\alpha}\right\}  \\
&  =\frac{\alpha}{\alpha-1}\log\text{Tr}\left\{  \phi_{B}^{\left(
2-\alpha\right)  /\alpha}\right\}  \\
&  =2\frac{1}{1-\left(  2-\alpha\right)  /\alpha}\log\text{Tr}\left\{
\phi_{B}^{\left(  2-\alpha\right)  /\alpha}\right\}  \\
&  =2H_{\left(  2-\alpha\right)  /\alpha}\left(  B\right)  _{\phi}\\
&  =2H_{\left(  2-\alpha\right)  /\alpha}\left(  A\right)  _{\phi}.
\end{align}
The first equality follows from applying (\ref{RMISibson}). The second
equality follows because $\phi_{AB}^{\alpha}=\phi_{AB}$ for a pure state
$\phi$. The third equality follows because $\phi_{A}^{1-\alpha}\left\vert
\phi\right\rangle _{AB}=\phi_{B}^{1-\alpha}\left\vert \phi\right\rangle _{AB}$
for a pure bipartite state $\phi$. The fourth equality follows by taking the
partial trace over system $A$. The rest of the equalities are straightforward,
by applying the definition of the R\'{e}nyi entropy.
\end{proof}

%discuss operational interpretation in terms of an error exponent for entanglement concentration. We consider  which implies that . The interpretation is there in (113) of quant-ph/0206097, Hayashi et al.

\begin{corollary}
The R\'{e}nyi squashed entanglement is normalized on maximally entangled
states, in the sense that for $\alpha\in(0,1)\cup(1,2]$
\begin{equation}
E_{\alpha}^{\operatorname{sq}}\left(  A;B\right)  _{\Phi}=\log d,
\end{equation}
where $d$ is the Schmidt rank of the maximally entangled state $\Phi_{AB}$.
\end{corollary}

\begin{proof}
This follows from applying Proposition~\ref{prop:squashed-pure} and from the
fact that the R\'enyi entropy is equal to $\log d$ for a maximally mixed state
of dimension $d$.
\end{proof}

\begin{table*}
\centering
\begin{tabular}{|c|c|c|}
\hline 
Property & Squashed Ent. of (\ref{eq:sqent}) & R\'enyi Squashed Ent. of (\ref{eq:renyisqent}) \tabularnewline
\hline 
\hline 
Normalized & \checkmark & \checkmark
\tabularnewline
& & for $\alpha\in(0,1)\cup(1,2]$
\tabularnewline
\hline 
LOCC Monotonicity & \checkmark & if Conj.~\ref{conj:monoA} true,
\tabularnewline
&&  then true for $\alpha\in(0,1)\cup(1,2]$
\tabularnewline
\hline 
Convexity & \checkmark & \checkmark 
\tabularnewline
&& for $\alpha\in \left(0,1\right)$
\tabularnewline
\hline
Faithfulness & \checkmark & vanishing on sep. states for $\alpha\in(0,1)\cup(1,\infty)$;
\tabularnewline
&& if \cite[Conj.~34]{BSW14} true,
\tabularnewline
&& then true for $\alpha\in(1,\infty)$
\tabularnewline
\hline 
Additivity & \checkmark & subadditive for $\alpha\in(0,1)\cup(1,\infty)$
\tabularnewline
\hline
Monogamy & \checkmark & ? \tabularnewline
\hline
Non-lockability & \XSolidBrush & ? \tabularnewline
\hline 
Asymptotic Continuity & \checkmark & ? \tabularnewline
\hline
\end{tabular}
\caption{Properties of the R\'enyi squashed entanglement in comparison to those of the original von Neumann entropy based squashed entanglement.
%The * indicates that the proof for this property relies on Conjecture~\ref{conj:monoA}, while the $\dagger$ indicates dependence on~\cite[Conjecture 34]{BSW14}.
The question marks indicate properties that remain open for the R\'enyi squashed entanglement of (\ref{eq:renyisqent}). Two of the properties for the R\'enyi squashed entanglement rely on conjectures. Conj.~\ref{conj:monoA} is the statement that the R\'enyi conditional quantum mutual information obeys monotonicity with respect to quantum operations on system $A$, i.e., for $\rho_{ABE}\in\mathcal{S}(\mathcal{H}_{ABE})$ and $\alpha\in(0,1)\cup(1,2]$, $I_{\alpha}\left(  A;B|E\right)  _{\rho}\geq I_{\alpha}\left(  A^{\prime};B|E\right)  _{\tau}$, where $\tau_{A^{\prime}BE} \equiv\mathcal{M}_{A\rightarrow A^{\prime}}(\rho_{ABE})$ and $\mathcal{M}_{A\rightarrow A^{\prime}}$ is a CPTP map acting on system $A$ alone. \cite[Conj.~34]{BSW14} is the statement that the R\'enyi conditional quantum mutual information is monotone non-decreasing in the R\'enyi parameter, i.e., $I_\alpha(A;B|C)\leq I_\beta(A;B|C)$ for $0\leq\alpha\leq\beta$.}
\label{sqentprops}
\end{table*}

\begin{proposition}
For $\rho_{AB}\in\mathcal{S}(\mathcal{H}_{AB})$, the R\'enyi squashed
entanglement is bounded from above by the R\'enyi entanglement of formation for
$\alpha\in\left(  0,1\right)  \cup\left(  1,2\right)  $ as
\begin{equation}
E_{\left(  2-\alpha\right)  /\alpha}^{F}\left(  A;B\right)  _{\rho}\geq
E_{\alpha}^{\operatorname{sq}}\left(  A;B\right)  _{\rho}.
\end{equation}
\end{proposition}

\begin{proof}
Consider an ensemble $\left\{  p_{X}\left(  x\right)  ,\left\vert \psi
^{x}\right\rangle _{AB}\right\}  $ which realizes $\rho_{AB}$, in the sense
that
\begin{equation}
\rho_{AB}=\sum_{x}p_{X}\left(  x\right)  \left\vert \psi^{x}\right\rangle
\left\langle \psi^{x}\right\vert _{AB}.
\end{equation}
We then form the classical-quantum state $\sigma_{XAB}$:
\begin{equation}
\sigma_{XAB}=\sum_{x}p_{X}\left(  x\right)  \left\vert x\right\rangle
\left\langle x\right\vert _{X}\otimes\left\vert \psi^{x}\right\rangle
\left\langle \psi^{x}\right\vert _{AB}. \label{eq:EoF-ext}
\end{equation}
Consider that
\begin{align}
H_{\beta}\left(  A|X\right)   &  =\frac{\beta}{1-\beta}\log\sum_{x}
p_{X}\left(  x\right)  \left[  \text{Tr}\left\{  \left(  \psi_{A}^{x}\right)
^{\beta}\right\}  ^{1/\beta}\right] \\
&  =\frac{1}{1-\beta}\log\left[  \sum_{x}p_{X}\left(  x\right)  \left[
\text{Tr}\left\{  \left(  \psi_{A}^{x}\right)  ^{\beta}\right\}  ^{1/\beta
}\right]  \right]  ^{\beta}.
\end{align}
For $\beta\in\left(  0,1\right)  \cup\left(  1,\infty\right)  $, we have from
concavity / convexity of the function $x^{\beta}$ that
\begin{equation}
\label{xtotheb}H_{\beta}\left(  A|X\right)  \geq\frac{1}{1-\beta}\log\sum
_{x}p_{X}\left(  x\right)  \text{Tr}\left\{  \left(  \psi_{A}^{x}\right)
^{\beta}\right\}  .
\end{equation}
Let $\beta=\left(  2-\alpha\right)  /\alpha$. Then, from (\ref{xtotheb}), for
$\alpha\in\left(  0,1\right)  \cup\left(  1,2\right)  $, we have that
\begin{equation}
\label{lhspart}H_{\left(  2-\alpha\right)  /\alpha}\left(  A|X\right)
\geq\frac{1}{1-\left(  2-\alpha\right)  /\alpha}\log\sum_{x}p_{X}\left(
x\right)  \text{Tr}\left\{  \left(  \psi_{A}^{x}\right)  ^{\left(
2-\alpha\right)  /\alpha}\right\}  .
\end{equation}

Now consider that
\begin{align}
\frac{1}{2}I_{\alpha}\left(  A;B|X\right)  _{\sigma} &  =\frac{1}{2}
\frac{\alpha}{\alpha-1}\log\sum_{x}p_{X}\left(  x\right)  2^{\left(
\frac{\alpha-1}{\alpha}\right)  I_{\alpha}\left(  A;B\right)  _{\psi^{x}}}\\
&  =\frac{1}{2}\frac{\alpha}{\alpha-1}\log\sum_{x}p_{X}\left(  x\right)
2^{\left(  \frac{\alpha-1}{\alpha}\right)  \frac{\alpha}{\alpha-1}
\log\text{Tr}\left\{  \left(  \psi_{A}^{x}\right)  ^{\left(  2-\alpha\right)
/\alpha}\right\}  }\\
&  =\frac{1}{2}\frac{\alpha}{\alpha-1}\log\sum_{x}p_{X}\left(  x\right)
\text{Tr}\left\{  \left(  \psi_{A}^{x}\right)  ^{\left(  2-\alpha\right)
/\alpha}\right\}  \\
&  =\frac{1}{1-\left(  2-\alpha\right)  /\alpha}\log\sum_{x}p_{X}\left(
x\right)  \text{Tr}\left\{  \left(  \psi_{A}^{x}\right)  ^{\left(
2-\alpha\right)  /\alpha}\right\}  .\label{rhspart}
\end{align}
The first equality follows by applying Lemma~\ref{lem:classical-cond}. The
second equality follows from some steps given in the proof of
Proposition~\ref{prop:squashed-pure}. The last two equalities are straightforward.

From (\ref{lhspart}) and (\ref{rhspart}), we have that
\begin{equation}
H_{\left(  2-\alpha\right)  /\alpha}\left(  A|X\right)  \geq\frac{1}
{2}I_{\alpha}\left(  A;B|X\right)  . \label{eq:cond-ent-CMI}
\end{equation}
The statement of the proposition follows from (\ref{eq:cond-ent-CMI}),
because
\begin{equation}
\frac{1}{2}I_{\alpha}\left(  A;B|X\right)  \geq E_{\alpha}^{\operatorname{sq}
}\left(  A;B\right)  _{\rho}, \label{eq:rel-renyi-squash}
\end{equation}
which in turn follows because the state in (\ref{eq:EoF-ext}) is a particular
extension of the state $\rho_{AB}$ and by definition, the R\'enyi squashed
entanglement is equal to the infimum over all such extensions. So putting
together (\ref{eq:cond-ent-CMI}) and (\ref{eq:rel-renyi-squash}) gives the
statement of the proposition.
\end{proof}

%%%%%%%%%%%%%%%%%%%%%%%%%%%%%%%%%%%%%%%%%%%%%%%%%%%%%%%%%%%%%%%%%%%%%%%%%%%%%%%%%

\section{R\'enyi Quantum Discord}

In this section, we define a R\'enyi quantum discord and explore some properties of it.
We also give an expression for the R\'enyi discord of pure bipartite states. Further, using the R\'enyi discord, we discuss an operational characterization of quantum states with small von Neumann entropy-based discord, which follows from \cite[Conjecture 34]{BSW14}.
\begin{definition}
\label{redis-def}
The R\'{e}nyi quantum discord of $\rho_{AB}\in\mathcal{S}(\mathcal{H})$ is defined for
$\alpha\in\left(  0,1)\cup(1,2\right]  $ as
\begin{equation}
D^{\alpha}(\overline{A};B)_{\rho}\equiv\inf_{\left\{  \Lambda_{x}\right\}
}I_{\alpha}\left(  E;B|X\right)  _{\omega},\label{renyidiscord}
\end{equation}
where
\begin{equation}
\omega_{EXB}\equiv U_{A\rightarrow EX}\rho_{AB}U_{A\rightarrow EX}^{\dag},
\end{equation}
with $U_{A\rightarrow EX}$ being an isometric extension of the measurement map
$\sigma\rightarrow\sum_{x}\operatorname{Tr}\left\{  \Lambda_{x}\sigma\right\}
\left\vert x\right\rangle \left\langle x\right\vert _{X}$ acting on system $A$.
\end{definition}

%%%%%%%%%%%%%%%%%%%%%%%%%%%%%%%%%%%%%%%%%%%%%%%%%%%%%%%%%%%%%%%%%%%%%%%%%%%%%%%%%

\subsection{Properties of the R\'enyi discord}

We now prove that the R\'enyi quantum discord for $\alpha\in\left(
0,1)\cup(1,2\right]  $ is non-negative, is invariant with respect to local unitaries,
vanishes on the set of classical-quantum states, and is optimized by a
rank-one POVM.

\begin{proposition}
The R\'enyi quantum discord $D^{\alpha}(\overline{A};B)_{\rho}$ is non-negative for $\alpha\in\left(0,1)\cup(1,2\right]$.
\end{proposition}

\begin{proof}
This follows easily from the fact that the R\'enyi conditional mutual
information is non-negative for $\alpha\in\left(  0,1)\cup(1,2\right]  $
\cite[Corollary 16]{BSW14}.
\end{proof}

\begin{proposition}
For a classical-quantum state $\rho_{AB}\in\mathcal{S}(\mathcal{H}_{AB})$ the R\'enyi quantum discord is equal to zero for $\alpha\in(0,1)\cup
(1,\infty)$.
\end{proposition}

\begin{proof}
Consider a classical-quantum state $\rho_{AB}$:
\begin{equation}
\rho_{AB}=\sum_{x}p_{X}\left(  x\right)  \left\vert x\right\rangle
\left\langle x\right\vert _{A}\otimes\rho_{B}^{x}.
\end{equation}
Let the dilation of a von Neumann measurement of system $A$ be $\left\vert x\right\rangle
_{A}\rightarrow\left\vert x\right\rangle _{X}\left\vert x\right\rangle _{E}$,
so that it produces
\begin{equation}
\rho_{XEB}\equiv\sum_{x}p\left(  x\right)  \left\vert x\right\rangle
\left\langle x\right\vert _{X}\otimes\left\vert x\right\rangle \left\langle
x\right\vert _{E}\otimes\rho_{B}^{x}.
\end{equation}
So the conditioning system $X$ is classical. Applying
Lemma~\ref{lem:classical-cond}, we find that
\begin{align}
I_{\alpha}\left(  E;B|X\right)=\frac{\alpha}{\alpha-1}\log\sum_{x} p_{X}\left(  x\right)  2^{\left(  \frac{\alpha-1}{\alpha}\right)  I_{\alpha}\left(  E;B\right)  _{\left\vert x\right\rangle \left\langle x\right\vert\otimes \rho^{x} }}=\frac{\alpha}{\alpha-1}\log\sum_{x}p_{X}\left(  x\right)=0.
\end{align}
Combining this with the previous proposition, we see that it is equal to zero
for classical-quantum states.
\end{proof}

\begin{proposition}
The R\'enyi quantum discord $D^{\alpha}(\overline{A};B)_{\rho}$ is invariant with respect to local unitaries for $\alpha
\in(0,1)\cup(1,\infty)$.
\end{proposition}

\begin{proof}
First, from the definition of $I_{\alpha}\left(  A;B|E\right)  $, we can see
that this quantity is invariant with respect to local unitaries. It then immediately
follows that the R\'enyi discord is invariant with respect to local unitaries
$U_{A}\otimes V_{B}$. That is, if the optimal measurement is given by
$\left\{  \Lambda_{A}^{x}\right\}  $, then the optimal measurement with respect to a
local unitary $U_{A}$ can be taken as $\left\{  U_{A}^{\dag}\Lambda_{A}
^{x}U_{A}\right\}  $.
\end{proof}

\begin{proposition}\label{prop:rank-one-discord} 
It suffices to optimize the R\'enyi quantum discord (Definition~\ref{redis-def}) with respect to rank-one POVMs
for $\alpha\in\left(  0,1\right)  \cup\left(  1,2\right]  $.
\end{proposition}

\begin{proof}
Consider an arbitrary measurement map
\begin{equation}
\mathcal{M}\left(  \sigma\right)  \equiv\sum_{x}\text{Tr}\left\{  \Lambda
_{x}\sigma\right\}  \left\vert x\right\rangle \left\langle x\right\vert _{X},
\end{equation}
and a spectral decomposition of each $\Lambda_{x}$:
\begin{equation}
\Lambda_{x}=\sum_{y}\mu_{xy}\left\vert \phi_{x,y}\right\rangle \left\langle
\phi_{x,y}\right\vert .
\end{equation}
For a fixed $x$, the set $\left\{  \left\vert \phi_{x,y}\right\rangle
\right\}  _{y}$ is orthonormal. Furthermore, the set $\left\{  \mu_{xy}\left\vert \phi_{x,y}\right\rangle \left\langle \phi_{x,y}\right\vert
\right\}  _{x,y}$ forms a POVM (a rank-one refinement of the original POVM).
We can then rewrite the measurement map as
\begin{equation}
\mathcal{M}\left(  \sigma\right)  =\sum_{x,y}\text{Tr}\left\{  \mu_{xy}\left\vert \phi_{x,y}\right\rangle \left\langle \phi_{x,y}\right\vert
\sigma\right\}  \left\vert x\right\rangle \left\langle x\right\vert _{X}.
\end{equation}
In this way, we can already see that there is some loss of information when
discarding the outcome $y$. An isometric extension of the original measurement
map is specified by
\begin{equation}
U_{A\rightarrow EX_{E}YX}^{\mathcal{M}}\left\vert \psi\right\rangle _{A}
\equiv\sum_{x,y}\sqrt{\mu_{xy}}\left(  \left\vert \phi_{x,y}\right\rangle
_{E}\left\langle \phi_{x,y}\right\vert _{A}\right)  \left\vert \psi
\right\rangle _{A}\otimes\left\vert x\right\rangle _{X_{E}}\left\vert
y\right\rangle _{Y}\otimes\left\vert x\right\rangle _{X}.
\end{equation}
This is because the original map is recovered by tracing over the
environmental systems $EX_{E}Y$:
\begin{align}
&  \text{Tr}_{EX_{E}Y}\left\{  U_{A\rightarrow EX_{E}YX}^{\mathcal{M}}
\sigma\left(  U_{A\rightarrow EX_{E}YX}^{\mathcal{M}}\right)  ^{\dag}\right\}
\nonumber\\
&  =\text{Tr}_{EX_{E}Y}\left\{  \sum_{x,y,x^{\prime},y^{\prime}}\sqrt{\mu_{xy}
}\sqrt{\mu_{x^\prime y^{\prime}}}\left(  \left\vert \phi_{x,y}\right\rangle
_{E}\left\langle \phi_{x,y}\right\vert _{A}\right)  \sigma\left\vert
\phi_{x^{\prime},y^{\prime}}\right\rangle _{A}\left\langle \phi_{x^{\prime
},y^{\prime}}\right\vert _{E}\otimes\left\vert x\right\rangle \left\langle
x^{\prime}\right\vert _{X_{E}}\otimes\left\vert y\right\rangle \left\langle
y^{\prime}\right\vert _{Y}\otimes\left\vert x\right\rangle \left\langle
x^{\prime}\right\vert _{X}\right\}  \\
&  =\sum_{x,y}\mu_{xy}\left(  \left\langle \phi_{x,y}\right\vert _{E}\left\vert
\phi_{x,y}\right\rangle _{E}\left\langle \phi_{x,y}\right\vert _{A}\right)
\sigma\left\vert \phi_{x,y}\right\rangle _{A}\otimes\left\vert x\right\rangle
\left\langle x\right\vert _{X}\\
&  =\sum_{x,y}\text{Tr}\left\{  \mu_{xy}\left\vert \phi_{x,y}\right\rangle
\left\langle \phi_{x,y}\right\vert \sigma\right\}  \left\vert x\right\rangle
\left\langle x\right\vert _{X}\\
&  =\mathcal{M}\left(  \sigma\right)  .
\end{align}
Let $\psi_{RAB}^{\rho}$ be a purification of the state $\rho_{AB}$ on which we
are evaluating the R\'{e}nyi discord. Then let $\omega_{REX_{E}YXB}$ be the
following pure state:
\begin{equation}
\omega_{REX_{E}YXB}\equiv U_{A\rightarrow EX_{E}YX}^{\mathcal{M}}\psi
_{RAB}^{\rho}\left(  U_{A\rightarrow EX_{E}YX}^{\mathcal{M}}\right)  ^{\dag}.
\end{equation}
Consider the following chain of inequalities:
\begin{align}
I_{\alpha}\left( EX_{E}Y;B|X\right)_{\omega}=I_{\alpha}\left(
B;EX_{E}Y|R\right)  _{\omega}\geq I_{\alpha}\left(  B;EX_{E}|R\right)  _{\omega}=I_{\alpha}\left(  EX_{E};B|XY\right)  _{\omega}.
\end{align}
The first equality is by duality of (\ref{eq:duality}). The second inequality is from
monotonicity with respect to local quantum operation. The last equality is again from duality. But now consider that the last quantity corresponds to the quantum discord for the following refined rank-one measurement map:
\begin{equation}
\sigma\rightarrow\sum_{x,y}\text{Tr}\left\{  \mu_{xy}\left\vert \phi
_{x,y}\right\rangle \left\langle \phi_{x,y}\right\vert \sigma\right\}
\left\vert x\right\rangle \left\langle x\right\vert _{X}\otimes\left\vert
y\right\rangle \left\langle y\right\vert _{Y}.
\end{equation}
Furthermore, the systems $EX_{E}$ above play the role of the environment of
the refined rank-one measurement map from $A$ to $X$ and $Y$. This is because
\begin{align}
&  \text{Tr}_{EX_{E}}\left\{  U_{A\rightarrow EX_{E}YX}^{\mathcal{M}}
\sigma\left(  U_{A\rightarrow EX_{E}YX}^{\mathcal{M}}\right)  ^{\dag}\right\}
\nonumber\\
&  =\text{Tr}_{EX_{E}}\left\{  \sum_{x,y,x^{\prime},y^{\prime}}\sqrt{\mu_{xy}
}\sqrt{\mu_{x^\prime y^{\prime}}}\left(  \left\vert \phi_{x,y}\right\rangle
_{E}\left\langle \phi_{x,y}\right\vert _{A}\right)  \sigma\left\vert
\phi_{x^{\prime},y^{\prime}}\right\rangle _{A}\left\langle \phi_{x^{\prime
},y^{\prime}}\right\vert _{E}\otimes\left\vert x\right\rangle \left\langle
x^{\prime}\right\vert _{X_{E}}\otimes\left\vert y\right\rangle \left\langle
y^{\prime}\right\vert _{Y}\otimes\left\vert x\right\rangle \left\langle
x^{\prime}\right\vert _{X}\right\}  \\
&  =\sum_{x,y,y^{\prime}}\mu_{xy}\left(  \left\langle \phi_{x,y^{\prime}
}\right\vert _{E}\left\vert \phi_{x,y}\right\rangle _{E}\left\langle
\phi_{x,y}\right\vert _{A}\right)  \sigma\left\vert \phi_{x,y^{\prime}
}\right\rangle _{A}\otimes\left\vert y\right\rangle \left\langle y^{\prime
}\right\vert _{Y}\otimes\left\vert x\right\rangle \left\langle x\right\vert
_{X}\\
&  =\sum_{x,y,y^{\prime}}\mu_{xy}\left\langle \phi_{x,y}\right\vert _{A}
\sigma\left\vert \phi_{x,y}\right\rangle _{A}\otimes\left\vert y\right\rangle
\left\langle y\right\vert _{Y}\otimes\left\vert x\right\rangle \left\langle
x\right\vert _{X}\\
&  =\sum_{x,y}\text{Tr}\left\{  \mu_{xy}\left\vert \phi_{x,y}\right\rangle
\left\langle \phi_{x,y}\right\vert \sigma\right\}  \left\vert x\right\rangle
\left\langle x\right\vert _{X}\otimes\left\vert y\right\rangle \left\langle
y\right\vert _{Y}.
\end{align}
Since a rank-one POVM\ always achieves a lower value of the R\'{e}nyi quantum
discord, it suffices to optimize over only these kinds of POVMs when
calculating it.
\end{proof}
\bigskip

At this point we would like to note that a different R\'enyi generalization of the quantum discord,
%$\mathcal{D}^{\alpha}(\overline{A}; B)_\rho$,
proposed in \cite{MBPDS14}, is optimized by a rank-one POVM as well. This quantity was defined
based upon an optimization over rank-one projective measurements, but
can be extended to an optimization over all POVMs. Consider the
definition for the R\'enyi discord of a bipartite state $\rho_{AB}$
proposed in \cite[Eq.~(26)]{MBPDS14} (with projective measurements replaced by POVMs $\{\Lambda^x\}$):
\begin{align}
\mathfrak{D}^{\alpha}(\overline{A}; B)_\rho &\equiv \min_{\sigma_A, \sigma_B} D_\alpha(\rho_{AB}\Vert \sigma_A\otimes\sigma_B)-\max_{\{\Lambda^x\}} \min_{\sigma_X, \sigma_B} D_\alpha(\rho_{XB}\Vert \sigma_X\otimes\sigma_B)\label{AP1}.
\end{align}
The sufficiency of rank-one POVMs $\{\Lambda^x\}$ in the above definition can be proven along the following lines. Similar to the proof of Proposition~\ref{prop:rank-one-discord}, consider a spectral decomposition of the POVM elements $\{\Lambda^x\}$:
\begin{equation}
\Lambda_{x}=\sum_{y}\mu_{xy}\left\vert \phi_{x,y}\right\rangle \left\langle
\phi_{x,y}\right\vert .
\end{equation}
Once again, for a fixed $x$, the set $\left\{  \left\vert \phi_{x,y}\right\rangle
\right\}  _{y}$ is orthonormal, and the set $\left\{  \mu
_{xy}\left\vert \phi_{x,y}\right\rangle \left\langle \phi_{x,y}\right\vert
\right\}  _{x,y}$ forms a rank-one refinement of the original POVM. The output state corresponding to the measurement consisting of these above rank-one refinements  can be written as
\begin{align}
\rho_{XYB}=\sum_{x,y}\left\vert x\right\rangle \left\langle x\right\vert _{X}\otimes \left\vert y\right\rangle \left\langle y\right\vert _{Y}\otimes \mu_{xy}\left\langle \phi_{x,y}\right\vert_A\rho_{AB}\left\vert \phi_{x,y}\right\rangle_A.
\end{align}
Due to the monotonicity of the R\'enyi relative entropy with respect to local quantum operations, we have that
\begin{align}
\min_{\sigma_{XY}, \sigma_B} D_\alpha(\rho_{XYB}\Vert \sigma_{XY}\otimes\sigma_B)  &\geq \min_{\sigma_X, \sigma_B} D_\alpha(\rho_{XB}\Vert \sigma_{X}\otimes\sigma_B).
\end{align}
This implies
\begin{align}
\max_{\left\{  \mu
_{y}\left\vert \phi_{x,y}\right\rangle \left\langle \phi_{x,y}\right\vert
\right\}  _{x,y}}\min_{\sigma_{XY}, \sigma_B} D_\alpha(\rho_{XYB}\Vert \sigma_{XY}\otimes\sigma_B)\geq \max_{\left\{  \mu
_{y}\left\vert \phi_{x,y}\right\rangle \left\langle \phi_{x,y}\right\vert
\right\}  _{x,y}}\min_{\sigma_X, \sigma_B} D_\alpha(\rho_{XB}\Vert \sigma_X\otimes\sigma_B).
\end{align}
Hence, the second term in (\ref{AP1}) can be replaced with 
\begin{equation}\max_{\left\{  \mu
_{y}\left\vert \phi_{x,y}\right\rangle \left\langle \phi_{x,y}\right\vert
\right\}  _{x,y}}\min_{\sigma_{XY}, \sigma_B} D_\alpha(\rho_{XYB}\Vert \sigma_{XY}\otimes\sigma_B).
\end{equation}
Therefore, it suffices to optimize the R\'{e}nyi quantum discord defined in~\cite[Eq.~(26)]{MBPDS14} over only rank-one POVMs.

\begin{table*}
\centering
\begin{tabular}{|c|c|c|c|}
\hline 
Property & Discord of (\ref{eq:sqent}) & R\'enyi Discord of~\eqref{renyidiscord}  & R\'enyi Discord of~\eqref{AP1}\tabularnewline
& & & as in~\cite{MBPDS14}\tabularnewline
\hline 
\hline 
Non-negative & \checkmark & \checkmark &
 \checkmark 
\tabularnewline
& & for $\alpha\in(0,1)\cup(1,2]$ & for $\alpha\in[1/2,1)\cup(1,\infty)$
\tabularnewline
\hline 
Vanishing on cq-states & \checkmark & \checkmark &
\checkmark 
\tabularnewline
& & for $\alpha\in(0,1)\cup(1,\infty)$ & for $\alpha\in(0,1)\cup(1,\infty)$
\tabularnewline
\hline 
Unitary invariance & \checkmark & \checkmark &
\checkmark
\tabularnewline
& & for $\alpha\in(0,1)\cup(1,\infty)$ & for $\alpha\in(0,1)\cup(1,\infty)$
\tabularnewline
\hline 
Rank-1 POVM optimal & \checkmark &  \checkmark &
\checkmark 
\tabularnewline
& & for $\alpha\in(0,1)\cup(1,2]$ & for $\alpha\in[1/2,1)\cup(1,\infty)$
\tabularnewline
\hline 
Monotone in $\alpha$ & N/A & if \cite[Conj. 34]{BSW14} & \XSolidBrush
\tabularnewline
&& true, then true &
\tabularnewline
\hline 
\end{tabular}
\caption{Properties of the R\'enyi quantum discord of (\ref{renyidiscord}) in comparison to those of the original von Neumann entropy based quantum discord and the R\'enyi quantum discord proposed in~\cite[Eq.~(26)]{MBPDS14}. In the first column, cq-states refers to classical-quantum states. The property of monotonicity in $\alpha$ for the R\'enyi quantum discord of (\ref{renyidiscord}) relies on \cite[Conj.~34]{BSW14}, which is the statement that the R\'enyi conditional quantum mutual information is monotone non-decreasing in the R\'enyi parameter, i.e., $I_\alpha(A;B|C)\leq I_\beta(A;B|C)$ for $0\leq\alpha\leq\beta$.
%the $\dagger$ indicates dependence on~\cite[Conjecture 34]{BSW14}.
The R\'enyi discord of \cite{MBPDS14} is not monotone in $\alpha$ because it is equal to a difference of two R\'enyi relative entropies.} 
%The * indicates that the proof for this property relies on Conjecture~\ref{conj:monoA}.
\label{discordprops}
\end{table*}

%%%%%%%%%%%%%%%%%%%%%%%%%%%%%%%%%%%%%%%%%%%%%%%%%%%%%%%%%%%%%%%%%%%%%%%%%%%%%%%%%

\subsection{R\'enyi quantum discord for pure bipartite states}

We now give an expression for the R\'{e}nyi discord of pure bipartite states.

\begin{proposition}
The R\'{e}nyi quantum discord of a pure bipartite state $\psi_{AB}\in\mathcal{S}(\mathcal{H}_{AB})$ for $\alpha\in(0,1)\cup(1,2]$ is given by
\begin{equation}
D^{\alpha}(  \overline{A};B)  _{\psi}=\inf_{\{\vert\varphi_x\rangle :
\sum_x \vert \varphi_x \rangle \langle \varphi_x \vert = I\}}\frac{\alpha}{\alpha-1}\log\sum_{x}p\left(  x\right) \left\langle\xi_x\right\vert_B\psi_B^{1-\alpha}\left\vert\xi_x\right\rangle^{1/\alpha},
%H_{1/\alpha}(B)_{\psi}.
\label{eq:purestdis}
\end{equation}
where 
\begin{align}\left\vert\xi_x\right\rangle_B=
\frac{\langle\varphi_x\vert_A\vert\psi\rangle_{AB}}{\sqrt{p(x)}},\quad p(x)=\Vert\langle\varphi_x\vert_A\vert\psi\rangle_{AB}\Vert_2^2,
\end{align}
and $\{\vert\varphi_x\rangle :
\sum_x \vert \varphi_x \rangle \langle \varphi_x \vert = I\}$ denotes a rank-one POVM acting on system $A$.
\end{proposition}

\begin{proof}
We begin by recalling Proposition~\ref{prop:rank-one-discord}, i.e., that it
suffices to optimize the R\'{e}nyi quantum discord over rank-one POVMs. Let
$\left\{  \left\vert \varphi_{x}\right\rangle \left\langle \varphi
_{x}\right\vert _{A}\right\}  _{x}$ denote such a POVM,\ so that $\sum
_{x}\left\vert \varphi_{x}\right\rangle \left\langle \varphi_{x}\right\vert
_{A}=I_{A}$. 
Consider a bipartite pure state
\begin{equation}
|\psi\rangle_{AB}=\sum_{y=0}^{d-1}\sqrt{\lambda(y)
}\left\vert \tilde{\psi}_y\right\rangle _{A}\left\vert \psi_y\right\rangle _{B},
\end{equation}
where $\left\vert \tilde{\psi}_y\right\rangle$ and $\left\vert \psi_y\right\rangle$ are orthonormal bases in $\mathcal{H}_A$ and $\mathcal{H}_B$.
The post measurement tripartite state is given by
\begin{equation}
\omega_{BEX}\equiv U_{A\rightarrow EX}\psi_{AB}U_{A\rightarrow EX}^{\dag},
\end{equation}
where $U_{A\rightarrow EX}$ is an isometric extension of the aforementioned
rank-one measurement:
\begin{equation}
U_{A\rightarrow EX}\equiv\sum_{x}\left\vert x\right\rangle _{E}\left\vert
x\right\rangle _{X}\left\langle \varphi_{x}\right\vert _{A}.
\end{equation}
The above state, and the reduced states on systems $BE$ and $B$ can thus be equivalently written as
\begin{align}
\label{purediseqns}
\omega_{BEX}&\equiv\sum_{x,y}\sqrt{p\left(  x\right)p\left(  y\right)}  \left\vert x\right\rangle
\left\langle y\right\vert _{X}\otimes\left\vert x\right\rangle \left\langle
y\right\vert _{E}\otimes\left\vert\xi_x\right\rangle\left\langle\xi_y\right\vert_B,\\
\omega_{BE}&\equiv\sum_{x}p\left(  x\right)\left\vert x\right\rangle \left\langle
x\right\vert _{E}\otimes\left\vert\xi_x\right\rangle\left\langle\xi_x\right\vert_B,
%\omega_{B}&\equiv\sum_{x}p\left(  x\right)\left\vert\xi_x\right\rangle\left\langle\xi_x\right\vert_B,
\end{align}
where $p(x)={\rm Tr}\{\vert\varphi_x\rangle\langle\varphi_x\vert_A\vert\psi\rangle\langle\psi\vert_{AB}\}$ and $\left\vert\xi_x\right\rangle\left\langle\xi_x\right\vert_B={\rm Tr}_A\{\vert\varphi_x\rangle\langle\varphi_x\vert_A\vert\psi\rangle\langle\psi\vert_{AB}\}/p(x)$.
The R\'{e}nyi conditional quantum mutual information of $\omega_{BEX}$ is thereby given by
\begin{align}
I_{\alpha}\left(  E;B|X\right)  _{\omega} &  =I_{\alpha}\left(  B;E\right)
_{\omega}\\
&  =\frac{\alpha}{\alpha-1}\log{\rm Tr}\left\{{\rm Tr}_B\left\{\omega_{BE}^\alpha\omega_B^{1-\alpha}\right\}^{1/\alpha}\right\}\\
&  =\frac{\alpha}{\alpha-1}\log{\rm Tr}\left\{{\rm Tr}_B\left\{\left(\sum_{x}p\left(  x\right)^\alpha  \left\vert x\right\rangle \left\langle
x\right\vert _{E}\otimes\left\vert\xi_x\right\rangle\left\langle\xi_x\right\vert_B\right)\omega_B^{1-\alpha}\right\}^{1/\alpha}\right\}\\
&  =\frac{\alpha}{\alpha-1}\log{\rm Tr}\left\{\sum_{x}p\left(  x\right) \left\vert x\right\rangle \left\langle
x\right\vert _{E} {\rm Tr}_B\left\{\left\vert\xi_x\right\rangle\left\langle\xi_x\right\vert_B\omega_B^{1-\alpha}\right\}^{1/\alpha}\right\}\\
&  =\frac{\alpha}{\alpha-1}\log{\rm Tr}\left\{\sum_{x}p\left(  x\right) {\rm Tr}_B\left\{\left\vert\xi_x\right\rangle\left\langle\xi_x\right\vert_B\omega_B^{1-\alpha}\right\}^{1/\alpha}\right\}\\
&  =\frac{\alpha}{\alpha-1}\log\sum_{x}p\left(  x\right) \left\langle\xi_x\right\vert_B\omega_B^{1-\alpha}\left\vert\xi_x\right\rangle^{1/\alpha}.\label{purerendis1}
\end{align}
The first equality follows from application of duality of (\ref{eq:duality})
along with the fact that $\omega_{BEX}$ is a pure state. The second and third equalities
follow from (\ref{RMISibson}), the fact that the system $E$ is classical, and
the fact that the post-measurement states on system $B$ are pure whenever a
rank-one POVM is performed on system $A$. The fourth one follows from tracing over the $E$ system. The fifth and sixth equalities are straightforward.
%The fourth equality follows from applying the trace over $E$ system first. (This can be done since $E$ is classical.) The expression in (\ref{purerendis1}) takes its least value when the states $\psi_B^x$ are orthogonal. This happens when the POVMs $\{\Lambda^x\}$ are in fact the orthogonal projectors corresponding to the Schmidt basis of the state $|\psi\rangle_{AB}$, and the value would be
%\begin{multline}
%\inf_{\left\{  \Lambda_{x}\right\}}\frac{\alpha}{\alpha-1}\log{\rm Tr}\left\{\left(\left(\sum_{x}p\left(  x\right)^\alpha \psi_{B}^{x}\right)\left(\sum_{x}p\left(  x\right)^{(1-\alpha)}\psi_{B}^{x}\right)\right)^{1/\alpha}\right\}\\
%= \frac{\alpha}{\alpha-1}\log\sum_{x}\lambda_x^{1/\alpha} {\rm Tr}\left\{\psi_{B}^{x}\right\}= \frac{\alpha}{\alpha-1}\log\sum_{x}\lambda_x^{1/\alpha}=H_{1/\alpha}(B)_{\psi}.
%\end{multline}
\end{proof}

\begin{corollary}
The R\'enyi quantum discord of a maximally entangled state $\Phi_{AB}\in\mathcal{S}(\mathcal{H}_{AB})$ simplifies to
\begin{equation}
D^{\alpha}\left(  \overline{A};B\right)  _{\Phi}=\log |A|.
\end{equation}
\end{corollary}

\begin{proof}
For a maximally entangled state, $p(x)$ in (\ref{purediseqns}) is equal to $1/|A|$, and the reduced state on system $B$ is maximally mixed. The result then follows from (\ref{purerendis1}).
\end{proof}

\subsection{Conjectured remainder terms for non-negativity of quantum discord}

In \cite[Conjecture
34]{BSW14}, it was conjectured that a R\'{e}nyi conditional quantum mutual information defined based upon the sandwiched R\'enyi relative entropy of (\ref{eq:def-sandwiched}) is also monotone in the R\'{e}nyi parameter. That is, for a
tripartite state $\rho_{ABC}$ and $0\leq\alpha\leq\beta$, it was conjectured
that
\begin{align}
\widetilde{I}_{\alpha}\left(  A;B|C\right)  _{\rho} &  \leq\widetilde
{I}_{\beta}\left(  A;B|C\right)  _{\rho},\label{eq:mono-alpha-2}
\end{align}
where the \textquotedblleft sandwiched\textquotedblright\ R\'{e}nyi conditional mutual information is defined as \cite[Section 6]{BSW14}
\begin{equation}
\widetilde{I}_{\alpha}\left(  A;B|C\right)  _{\rho}\equiv\frac{1}
{\alpha-1}\log\left\Vert \rho_{ABC}^{1/2}\rho_{AC}^{(1-\alpha)/2\alpha}
\rho_{C}^{(\alpha-1)/2\alpha}\rho_{BC}^{(1-\alpha)/2\alpha}\right\Vert
_{2\alpha}^{2\alpha}.
\end{equation}
Proofs were given for this conjectured inequality in \cite[Section 8]{BSW14} in some special cases, e.g.,
when the R\'{e}nyi parameter $\alpha$ is in a neighborhood of one, and when $1/\alpha+1/\beta=2$. Furthermore, implications of the conjectured inequality for understanding states with small conditional quantum mutual
information were discussed. In particular, it was pointed out that the following lower bounds on the conditional quantum mutual
information hold as consequences of the conjectured inequality in (\ref{eq:mono-alpha-2}):
%and the fact that the following inequality holds for $\alpha\in(0,1)\cup
%(1,\infty)$ \cite{DL13}:%
%\begin{equation}
%\widetilde{D}_{\alpha}\left(  \rho\Vert\sigma\right)  \leq D_{\alpha}\left(
%\rho\Vert\sigma\right),
%\end{equation}
\begin{align}
I\left(  A;B|C\right)  _{\rho} &  \geq -\log F\left(  \rho_{ABC},\mathcal{R}_{C\rightarrow AC}^{P}\left(\rho_{BC}\right)\right),\label{pmap1}\\
%\geq\frac{1}{4}\left\Vert \rho_{ABC}-\mathcal{R}_{C\rightarrow AC}^{P}\left(  \rho_{BC}\right)  \right\Vert_{1}^{2}
I\left(  A;B|C\right)  _{\rho} &\geq -\log F\left(  \rho_{ABC},\mathcal{R}_{C\rightarrow BC}^{P}\left(\rho_{AC}\right)  \right),\label{pmap2}
%\geq\frac{1}{4}\left\Vert \rho_{ABC}-\mathcal{R}_{C\rightarrow BC}^{P}\left(  \rho_{AC}\right)  \right\Vert_{1}^{2}
\end{align}
where $\mathcal{R}_{C\rightarrow AC}^{P}$ and $\mathcal{R}_{C\rightarrow
BC}^{P}$ are Petz recovery maps~\cite{HJPW04}:
\begin{align}
\mathcal{R}_{C\rightarrow AC}^{P}(\cdot) &  \equiv\rho_{AC}^{1/2}\rho
_{C}^{-1/2}(\cdot)\rho_{C}^{-1/2}\rho_{AC}^{1/2},\\
\mathcal{R}_{C\rightarrow BC}^{P}(\cdot) &  \equiv\rho_{BC}^{1/2}\rho
_{C}^{-1/2}(\cdot)\rho_{C}^{-1/2}\rho_{BC}^{1/2}.
\end{align}
%Similar bounds were conjectured previously in~\cite{K13conj}.

Since the quantum discord is based upon the conditional quantum mutual
information, we now examine the implications of this conjecture for the
quantum discord, by writing down a corresponding lower bound on it. This lower bound if true would provide a characterization for states with small quantum discord (the von Neumann entropy based quantity that is). It has the interpretation of quantifying how far a quantum state is
from being a fixed point of an entanglement-breaking channel. In case the discord
is equal to zero we can conclude that the state is a fixed point of an entanglement-breaking channel. In this case, we can apply~\cite[Theorem 5.3]{FNW14} to conclude the known result that any zero-discord state is in fact a classical-quantum state.

\begin{consequence}
Assuming \cite[Conjecture 34]{BSW14}, the following lower bound holds for the
quantum discord of $\rho_{AB}\in\mathcal{S}(\mathcal{H}_{AB})$:
\begin{align}
%D(  \overline{A};B)  _{\rho} &  \geq\min_{\left\{
%\left\vert \phi_{x}\right\rangle :\sum_{x}\left\vert \phi_{x}\right\rangle
%\left\langle \phi_{x}\right\vert =I\right\}  }
%-\log F\left(\rho_{AB},\sum_{x,y}\operatorname{Tr}\left\{  \left\vert \phi_{y}\right\rangle \left\langle\phi_{x}\right\vert _{A}\rho_{AB}\right\}  \left\vert \phi_{x}\right\rangle
%\left\langle \phi_{y}\right\vert _{A}\otimes\left(  \rho_{B}^{x}\right)
%^{1/2}\left(  \rho_{B}^{y}\right)  ^{1/2}\right),\label{discbound1}\\
D(  \overline{A};B)  _{\rho} &  \geq\min_{\left\{
\left\vert \phi_{x}\right\rangle :\sum_{x}\left\vert \phi_{x}\right\rangle
\left\langle \phi_{x}\right\vert =I\right\}  }
-\log F\left(\rho_{AB},\mathcal{E}_A(\rho_{AB})\right),\label{discbound2}
\end{align}
where 
\begin{equation}\rho_{B}^{x}\equiv\frac{1}{\operatorname{Tr}\left\{  \left\vert \phi_{x}\right\rangle
\left\langle \phi_{x}\right\vert _{A}\rho_{AB}\right\}  }\left\langle \phi
_{x}\right\vert _{A}\rho_{AB}\left\vert \phi_{x}\right\rangle _{A},\end{equation}
and $\mathcal{E}_A$ is the following entanglement-breaking channel:
\begin{equation}
\mathcal{E}_A(\sigma_A) \equiv \sum_{x}
\left\langle \phi_{x}\right\vert _{A}\sigma_{A}\left\vert \phi_{x}
\right\rangle _{A}
\frac{\rho_{A}^{1/2}\left\vert \phi_{x}\right\rangle \left\langle \phi
_{x}\right\vert _{A}\rho_{A}^{1/2}}{\operatorname{Tr}\left\{  \left\vert \phi
_{x}\right\rangle \left\langle \phi_{x}\right\vert _{A}\rho_{A}\right\}
}.
\end{equation}
\label{cons:discord}
\end{consequence}

\begin{proof}
Consider a rank-one measurement $\left\{  \left\vert \phi_{x}\right\rangle
\left\langle \phi_{x}\right\vert \right\}  $, and its isometric extension
\begin{equation}
U_{A\rightarrow XE}=\sum_{x}\left\vert x\right\rangle _{X}\left\vert
x\right\rangle _{E}\left\langle \phi_{x}\right\vert _{A}.
\end{equation} 
For a given state $\rho_{AB}$, the state relevant for discord becomes
\begin{align}
U_{A\rightarrow XE}\rho_{AB}U_{A\rightarrow XE}^{\dag} &  =\sum_{x,y}
\left\vert x\right\rangle _{X}\left\vert x\right\rangle _{E}\left\langle
\phi_{x}\right\vert _{A}\rho_{AB}\left\vert \phi_{y}\right\rangle
_{A}\left\langle y\right\vert _{X}\left\langle y\right\vert _{E}\\
&  =\sum_{x,y}\left\langle \phi_{x}\right\vert _{A}\rho_{AB}\left\vert
\phi_{y}\right\rangle _{A}\otimes\left\vert x\right\rangle \left\langle
y\right\vert _{X}\otimes\left\vert x\right\rangle \left\langle y\right\vert
_{E}.
\end{align}
The lower bound in (\ref{pmap1}) for this state takes the form:
\begin{equation}
I\left(  E;B|X\right)  \geq-\log F\left(\rho_{BEX},\rho_{EX}
^{1/2}\rho_{X}^{-1/2}\rho_{BX}\rho_{X}^{-1/2}\rho_{EX}^{1/2}\right).\label{eq:discord-ineq}
\end{equation}
So, we need to calculate $\rho_{BX}$, $\rho_{X}$, and $\rho_{EX}$:
\begin{align}
\rho_{BX} &  =\sum_{x}\left\langle \phi_{x}\right\vert _{A}\rho_{AB}\left\vert
\phi_{x}\right\rangle _{A}\otimes\left\vert x\right\rangle \left\langle
x\right\vert _{X},\\
\rho_{X} &  =\sum_{x}\text{Tr}\left\{  \left\vert \phi_{x}\right\rangle
\left\langle \phi_{x}\right\vert _{A}\rho_{AB}\right\}  \left\vert
x\right\rangle \left\langle x\right\vert _{X},\\
\rho_{EX} &  =\sum_{x,y}\text{Tr}\left\{  \left\vert \phi_{y}\right\rangle
\left\langle \phi_{x}\right\vert _{A}\rho_{AB}\right\}  \left\vert
x\right\rangle \left\langle y\right\vert _{X}\otimes\left\vert x\right\rangle
\left\langle y\right\vert _{E}.
\end{align}
Let
\begin{equation}
\rho_{B}^{x}\equiv\frac{1}{\text{Tr}\left\{  \left\vert \phi_{x}\right\rangle
\left\langle \phi_{x}\right\vert _{A}\rho_{AB}\right\}  }\left\langle \phi
_{x}\right\vert _{A}\rho_{AB}\left\vert \phi_{x}\right\rangle _{A}.
\end{equation}
Then we have that
\begin{align}
&  \rho_{EX}^{1/2}\rho_{X}^{-1/2}\rho_{BX}\rho_{X}^{-1/2}\rho_{EX}
^{1/2}\nonumber\\
&  =\rho_{EX}^{1/2}\left(  \sum_{x}\left\vert x\right\rangle \left\langle
x\right\vert _{X}\otimes\rho_{B}^{x}\right)  \rho_{EX}^{1/2}\\
&  =\left(  U_{A\rightarrow XE}\rho_{A}U_{A\rightarrow XE}^{\dag}\right)
^{1/2}\left(  \sum_{x}\left\vert x\right\rangle \left\langle x\right\vert
_{X}\otimes\rho_{B}^{x}\right)  \left(  U_{A\rightarrow XE}\rho_{A}
U_{A\rightarrow XE}^{\dag}\right)  ^{1/2}\\
&  =U_{A\rightarrow XE}\rho_{A}^{1/2}U_{A\rightarrow XE}^{\dag}\left(
\sum_{x}\left\vert x\right\rangle \left\langle x\right\vert _{X}\otimes
\rho_{B}^{x}\right)  U_{A\rightarrow XE}\rho_{A}^{1/2}U_{A\rightarrow
XE}^{\dag}.
\end{align}
Sandwiching by $U_{A\rightarrow XE}^{\dag}\left(  \cdot\right)
U_{A\rightarrow XE}$ then gives
\begin{align}
&  \rho_{A}^{1/2}U_{A\rightarrow XE}^{\dag}\left(  \sum_{x}\left\vert
x\right\rangle \left\langle x\right\vert _{X}\otimes\rho_{B}^{x}\right)
U_{A\rightarrow XE}\rho_{A}^{1/2}\nonumber\\
&  =\rho_{A}^{1/2}\left(  \sum_{z}\left\vert \phi_{z}\right\rangle
_{A}\left\langle z\right\vert _{X}\left\langle z\right\vert _{E}\right)
\left(  \sum_{x}\left\vert x\right\rangle \left\langle x\right\vert
_{X}\otimes\rho_{B}^{x}\right)  \left(  \sum_{z^{\prime}}\left\vert z^{\prime
}\right\rangle _{X}\left\vert z^{\prime}\right\rangle _{E}\left\langle
\phi_{z^{\prime}}\right\vert _{A}\right)  \rho_{A}^{1/2}\\
&  =\rho_{A}^{1/2}\left(  \sum_{z,x,z^{\prime}}\left\vert \phi_{z}
\right\rangle _{A}\left\langle \phi_{z^{\prime}}\right\vert _{A}\left\langle
z\right\vert \left\vert z^{\prime}\right\rangle _{E}\left\langle z\right\vert
\left\vert x\right\rangle \left\langle x\right\vert \left\vert z^{\prime
}\right\rangle _{X}\otimes\rho_{B}^{x}\right)  \rho_{A}^{1/2}\\
&  =\rho_{A}^{1/2}\left(  \sum_{x}\left\vert \phi_{x}\right\rangle
\left\langle \phi_{x}\right\vert _{A}\otimes\rho_{B}^{x}\right)  \rho
_{A}^{1/2}\\
&  =\sum_{x}\rho_{A}^{1/2}\left\vert \phi_{x}\right\rangle \left\langle
\phi_{x}\right\vert _{A}\rho_{A}^{1/2}\otimes\rho_{B}^{x},
\end{align}
which we can see is a density operator on systems $A$ and $B$.
This establishes that $\rho_{EX}^{1/2}\rho_{X}^{-1/2}\rho_{BX}\rho_{X}^{-1/2}\rho_{EX}
^{1/2}$ is in the subspace onto which
$U_{A\rightarrow XE} U_{A\rightarrow XE}^{\dag}$ projects,
and since this is true also for the state $\rho_{BXE}$, we find that
\begin{equation}
F\left(\rho_{BEX},\rho_{EX}
^{1/2}\rho_{X}^{-1/2}\rho_{BX}\rho_{X}^{-1/2}\rho_{EX}^{1/2}\right)
=F\left(\rho_{AB},\sum_{x}\rho_{A}^{1/2}\left\vert \phi_{x}\right\rangle \left\langle
\phi_{x}\right\vert _{A}\rho_{A}^{1/2}\otimes\rho_{B}^{x}\right)
\end{equation}
which is the bound in (\ref{discbound2}).
\end{proof}

%%%%%%%%%%%%%%%%%%%%%%%%%%%%%%%%%%%%%%%%%%%%%%%%%%%%%%%%%%%%%%%%%%%%%%%%%%%%%%%%%

\section{R\'enyi generalizations of relative entropy differences}

In our prior work, we detailed a procedure to obtain R\'{e}nyi generalizations of
quantum information measures that are equal to linear combinations of von
Neumann entropies with coefficients chosen from the set $\{-1,0,1\}$. The
procedure relied on writing any given measure as a relative entropy and
subsequently using the generalized Lie-Trotter product formula \cite{S85} and
the R\'{e}nyi relative entropy \cite{P86,MDSFT13,WWY13} or the sandwiched
R\'{e}nyi relative entropy \cite{MDSFT13,WWY13}. In this section, we show that
the same procedure can be used to obtain R\'{e}nyi generalizations of quantum
information measures that are equal to differences of relative entropies. Further, we conjecture that the proposed R\'enyi generalizations of relative entropy differences are monotone in the R\'enyi parameter. Based on this conjecture, we then suggest remainder terms for two central properties of the relative entropy, namely monotonicity with respect to quantum operations and joint convexity. We also conjecture a remainder term for the Holevo bound.

Consider the following difference of relative entropies:
\begin{equation}
D\left(  \rho\Vert\sigma\right)  -D\left(  \mathcal{N}\left(  \rho\right)
\Vert\mathcal{N}\left(  \sigma\right)  \right)  ,\label{eq:rel-entropy-diff}
\end{equation}
where $\rho,\sigma\in\mathcal{S}\left(  \mathcal{H}\right)  _{++}$, and
$\mathcal{N}$ is a strict CPTP map. The above measure is non-negative, which
is a consequence of the monotonicity of relative entropy \cite{L74,Lindblad1975}. Writing
the adjoint of a CPTP map $\mathcal{N}$ as $\mathcal{N}^{\dag}$, consider that
the relative entropy difference in (\ref{eq:rel-entropy-diff}) can be written
as \cite{CL14,Z14}
\begin{align}
&  D\left(  \rho\Vert\sigma\right)  -D\left(  \mathcal{N}\left(  \rho\right)
\Vert\mathcal{N}\left(  \sigma\right)  \right)  \nonumber\\
&  =\text{Tr}\left\{  \rho\left[  \log\rho-\log\sigma\right]  \right\}
-\text{Tr}\left\{  \mathcal{N}\left(  \rho\right)  \left[  \log\mathcal{N}
\left(  \rho\right)  -\log\mathcal{N}\left(  \sigma\right)  \right]  \right\}
\\
&  =\text{Tr}\left\{  \rho\left[  \log\rho-\log\sigma\right]  \right\}
-\text{Tr}\left\{  \rho\left[  \mathcal{N}^{\dag}\left(  \log\mathcal{N}
\left(  \rho\right)  \right)  -\mathcal{N}^{\dag}\left(  \log\mathcal{N}
\left(  \sigma\right)  \right)  \right]  \right\}  \\
&  =\text{Tr}\left\{  \rho\log\rho\right\}  -\text{Tr}\left\{  \rho\left[
\log\sigma+\mathcal{N}^{\dag}\left(  \log\mathcal{N}\left(  \rho\right)
-\log\mathcal{N}\left(  \sigma\right)  \right)  \right]  \right\} \label{vonNreldiff} \\
&  =D\left(  \rho\middle\Vert\exp\left\{  \log\sigma+\mathcal{N}^{\dag}\left(
\log\mathcal{N}\left(  \rho\right)  -\log\mathcal{N}\left(  \sigma\right)
\right)  \right\}  \right).\label{rengenreldif1}
\end{align}
In order to find R\'{e}nyi generalizations of the above quantity, we consider
the following lemma:

\begin{lemma}\label{LTBoost}
Let $\rho,\sigma\in\mathcal{S}\left(  \mathcal{H}\right)
_{++}$ and $\mathcal{N}$ be a strict CPTP map. Then
\begin{multline}
\lim_{p\rightarrow0}\left[  \sigma^{p/2}\mathcal{N}^{\dag}\left(  \left[
\mathcal{N}\left(  \sigma\right)  \right]  ^{-p/2}\left[  \mathcal{N}\left(
\rho\right)  \right]  ^{p}\left[  \mathcal{N}\left(  \sigma\right)  \right]
^{-p/2}\right)  \sigma^{p/2}\right]  ^{1/p}\label{eq:limit-renyi-gen}\\
=\exp\left\{  \log\sigma+\mathcal{N}^{\dag}\left(  \log\mathcal{N}\left(
\rho\right)  -\log\mathcal{N}\left(  \sigma\right)  \right)  \right\} .
\end{multline}
\end{lemma}

\begin{proof}
The method of proof is the same as that given for \cite[Eq.~(2.2)]{AH12}.
Consider that
\begin{align}
& \sigma^{p/2}\mathcal{N}^{\dag}\left(  \left[  \mathcal{N}\left(
\sigma\right)  \right]  ^{-p/2}\left[  \mathcal{N}\left(  \rho\right)
\right]  ^{p}\left[  \mathcal{N}\left(  \sigma\right)  \right]  ^{-p/2}
\right)  \sigma^{p/2}\nonumber\\
& =\left(  I+\frac{p}{2}\log\sigma+o\left(  p\right)  \right)  \times
\nonumber\\
& \ \ \ \ \ \ \ \mathcal{N}^{\dag}\left(  \left[  I-\frac{p}{2}
\log\mathcal{N}\left(  \sigma\right)  +o\left(  p\right)  \right]  \left[
I+p\log\mathcal{N}\left(  \rho\right)  +o\left(  p\right)  \right]  \left[
I-\frac{p}{2}\log\mathcal{N}\left(  \sigma\right)  +o\left(  p\right)
\right]  \right)  \times\nonumber\\
& \ \ \ \ \ \ \ \left(  I+\frac{p}{2}\log\sigma+o\left(  p\right)  \right)
\\
& =I+p\log\sigma-p\mathcal{N}^{\dag}\left(  \log\mathcal{N}\left(
\sigma\right)  \right)  +p\mathcal{N}^{\dag}\left(  \log\mathcal{N}\left(
\rho\right)  \right)  +o\left(  p\right)  \\
& =I+p\left[  \log\sigma+\mathcal{N}^{\dag}\left(  \log\mathcal{N}\left(
\rho\right)  -\log\mathcal{N}\left(  \sigma\right)  \right)  \right]
+o\left(  p\right)  ,
\end{align}
where $o\left(  p\right)  $ means that $o\left(  p\right)  /p\rightarrow0$ in
the operator norm as $p\rightarrow0$. The first equality follows from the
Taylor expansion $X^{p}=I+p\log X+o\left(  p\right)  $, and the second
equality follows because $\mathcal{N}^{\dag}\left(  I\right)  =I$. So this
implies for $p$ small enough that
\begin{align}
& \frac{1}{p}\log\left[  \sigma^{p/2}\mathcal{N}^{\dag}\left(  \left[
\mathcal{N}\left(  \sigma\right)  \right]  ^{-p/2}\left[  \mathcal{N}\left(
\rho\right)  \right]  ^{p}\left[  \mathcal{N}\left(  \sigma\right)  \right]
^{-p/2}\right)  \sigma^{p/2}\right]  \nonumber\\
& =\frac{1}{p}\log\left[  I+p\left[  \log\sigma+\mathcal{N}^{\dag}\left(
\log\mathcal{N}\left(  \rho\right)  -\log\mathcal{N}\left(  \sigma\right)
\right)  \right]  +o\left(  p\right)  \right]  \\
& =\log\sigma+\mathcal{N}^{\dag}\left(  \log\mathcal{N}\left(  \rho\right)
-\log\mathcal{N}\left(  \sigma\right)  \right)  +o\left(  p\right)  .
\end{align}
By exponentiating the first and last line and taking the limit as
$p\rightarrow0$, we recover the statement of the lemma.
\end{proof}

Appealing to Lemma~\ref{LTBoost}, we can write the second argument of the
relative entropy in (\ref{rengenreldif1}) as in (\ref{eq:limit-renyi-gen}).
Based on the procedure given in~\cite[Section 10]{BSW14}, this suggests the
following R\'{e}nyi generalizations of the relative entropy difference in
(\ref{eq:rel-entropy-diff}) using the R\'enyi relative entropy and the sandwiched R\'enyi relative entropy.

\begin{definition}\label{rengen-reldif}
Let $\rho,\sigma\in\mathcal{S}\left(  \mathcal{H}\right)
_{++}$ and $\mathcal{N}$ be a strict CPTP map. The R\'enyi generalizations of the relative entropy difference in \eqref{eq:rel-entropy-diff} are defined respectively as
\begin{align}
&\Delta_{\alpha}(\rho,\sigma,\mathcal{N})\equiv\frac{1}{\alpha-1}\log
\operatorname{Tr}\left\{  \rho^{\alpha}\sigma^{\left(  1-\alpha\right)  /2}
\mathcal{N}^{\dag}\left(  \left[  \mathcal{N}\left(  \sigma\right)  \right]
^{\left(  \alpha-1\right)  /2}\left[  \mathcal{N}\left(  \rho\right)  \right]
^{1-\alpha}\left[  \mathcal{N}\left(  \sigma\right)  \right]  ^{\left(
\alpha-1\right)  /2}\right)  \sigma^{\left(  1-\alpha\right)  /2}\right\},
\label{eq:Delta-a}\\
&\widetilde{\Delta}_{\alpha}(\rho,\sigma,\mathcal{N})\equiv\nonumber\\
&\frac{\alpha}{\alpha-1}\log\operatorname{Tr}\left\Vert \rho^{1/2}\sigma^{\left(
1-\alpha\right)  /2\alpha}\mathcal{N}^{\dag}\left(  \left[  \mathcal{N}\left(
\sigma\right)  \right]  ^{\left(  \alpha-1\right)  /2\alpha}\left[
\mathcal{N}\left(  \rho\right)  \right]  ^{\left(  1-\alpha\right)  /\alpha
}\left[  \mathcal{N}\left(  \sigma\right)  \right]  ^{\left(  \alpha-1\right)
/2\alpha}\right)  \sigma^{\left(  1-\alpha\right)  /2\alpha}\rho
^{1/2}\right\Vert _{\alpha}.\label{eq:Delta-SW-a}
\end{align}
\end{definition}

Both of these quantities converge to (\ref{eq:rel-entropy-diff}) in the limit
as $\alpha\rightarrow1$, which follows from an argument similar to the proof of \cite[Theorems 9 and 20]{BSW14}
(alternatively, one could employ Lemma~\ref{LTBoost},
\eqref{rengenreldif1}, and the fact that the R\'enyi and sandwiched R\'enyi relative entropies converge to $D(\rho\Vert \sigma)$ as $\alpha \to 1$ and 
are continuous with respect to the second argument
when it is positive definite).
We now show that the above R\'enyi generalizations of a difference of
two relative entropies are consistent with the R\'enyi conditional mutual
information proposed earlier in~\cite[Sections 5 and 6]{BSW14}.
Recall that the conditional mutual information can be written as
\begin{align}
I\left(  A;B|C\right)   &  =I\left(  B;AC\right)  -I(B;C)\\
&  =D\left(  \rho_{ABC}\Vert\rho_{B}\otimes\rho_{AC}\right)  -D\left(
\rho_{BC}\Vert\rho_{B}\otimes\rho_{C}\right),
\end{align}
and choosing in \eqref{eq:Delta-a} and \eqref{eq:Delta-SW-a},
\begin{align}
\rho =\rho_{ABC},\quad\sigma =\rho_{B}\otimes\rho_{AC},\quad\mathcal{N}\left(  \cdot\right)  =\operatorname{Tr}_{A}\left\{
\cdot\right\},\label{consisprope}
\end{align}
so that $\mathcal{N}^{\dag}\left(  \cdot\right)  =\left(  \cdot\right)  \otimes
I_{A}$,
leads to
\begin{align}
\Delta_{\alpha}\left(  \rho_{ABC},\rho_{B}\otimes\rho_{AC},\text{Tr}
_{A}\left\{  \cdot\right\}  \right)=\frac{1}{\alpha-1}\log\text{Tr}\left\{  \rho_{ABC}^{\alpha}\rho_{B}^{0}\rho_{AC}^{\left(  1-\alpha\right)  /2}\rho_{C}^{\left(  \alpha-1\right)
/2}\rho_{BC}^{1-\alpha}\rho_{C}^{\left(  \alpha-1\right)  /2}\rho
_{AC}^{\left(  1-\alpha\right)  /2}\rho_{B}^{0}\right\}  , \label{consis1}
\end{align}
as well as
\begin{align}
&  \widetilde{\Delta}_{\alpha}\left(  \rho_{ABC},\rho_{B}\otimes\rho
_{AC},\text{Tr}_{A}\left\{  \cdot\right\}  \right) \nonumber\\
&  =\frac{\alpha}{\alpha-1}\log\left\Vert \rho_{ABC}^{1/2}\rho_{B}^{0}
\rho_{AC}^{\left(  1-\alpha\right)  /2\alpha}\rho_{C}^{\left(  \alpha
-1\right)  /2\alpha}\rho_{BC}^{\left(  1-\alpha\right)  /\alpha}\rho
_{C}^{\left(  \alpha-1\right)  /2\alpha}\rho_{AC}^{\left(  1-\alpha\right)
/2\alpha}\rho_{B}^{0}\rho_{ABC}^{1/2}\right\Vert _{\alpha}. \label{consis2}
\end{align}
We conclude that the expressions in (\ref{consis1}) and (\ref{consis2}) are equal to the
R\'enyi generalizations of the conditional quantum mutual information proposed
in \cite{BSW14}.

%%%%%%%%%%%%%%%%%%%%%%%%%%%%%%%%%%%%%%%%%%%%%%%%%%%%%%%%%%%%%%%%%%%%%%%%%%%%%%%%%

\subsection{Conjectured remainder term for monotonicity of relative entropy}

\cite[Conjecture 34]{BSW14} leads us to one more conjecture, that the above generalizations $\Delta
_{\alpha}(\rho,\sigma,\mathcal{N})$ and $\widetilde{\Delta}_{\alpha}(\rho,\sigma,\mathcal{N})$ are monotone in the R\'{e}nyi parameter as well:

\begin{conjecture}\label{conj:mono-rel-diff}
Let $\rho,\sigma\in\mathcal{S}\left(  \mathcal{H}\right)  _{++}$ and $\mathcal{N}$ be a strict CPTP\ map. Then the following
inequalities hold for $0\leq\alpha\leq\beta$:
\begin{align}
\Delta_{\alpha}(\rho,\sigma,\mathcal{N}) &  \leq\Delta_{\beta}(\rho
,\sigma,\mathcal{N})\label{monreldif1},\\
\widetilde{\Delta}_{\alpha}(\rho,\sigma,\mathcal{N}) &  \leq\widetilde{\Delta
}_{\beta}(\rho,\sigma,\mathcal{N})\label{monreldif2}.
\end{align}
\end{conjecture}

When $\rho$, $\sigma\in \mathcal{S}(\mathcal{H})_{++}$, and $\mathcal{N}$ is a strict CPTP map, it is possible prove Conjecture \ref{conj:mono-rel-diff} in some special cases. For example, when $\alpha$ is in a neighborhood of one, the derivatives of $\Delta_{\alpha}(\rho,\sigma,\mathcal{N}) $ and $\widetilde{\Delta}_{\alpha}(\rho,\sigma,\mathcal{N})$ with respect to $\alpha$ are non-negative, as is seen using Taylor expansions of these functions (see Appendix \ref{renreldifApp} for details). This implies that monotonicity holds when $\alpha$ is in a neighborhood of one. Further, when $\alpha$ and $\beta$ satisfy $0\leq\alpha\leq\beta$, $\alpha+\beta=2$, and $\mathcal{N}$ is a unitary quantum channel, the inequality in (\ref{monreldif1}) for $\Delta_{\alpha}(\rho,\sigma,\mathcal{N}) $ is true. This is because, in such a
case, we have that $1-\alpha=-\left(  1-\beta\right)  $, so that
\begin{multline}
\left[\sigma^{\left(  1-\alpha\right)  /2}
\mathcal{N}^{\dag}\left(  \left[  \mathcal{N}\left(  \sigma\right)  \right]
^{\left(  \alpha-1\right)  /2}\left[  \mathcal{N}\left(  \rho\right)  \right]
^{1-\alpha}\left[  \mathcal{N}\left(  \sigma\right)  \right]  ^{\left(
\alpha-1\right)  /2}\right)  \sigma^{\left(  1-\alpha\right)  /2}\right]^{1/(1-\alpha)}\\
=\left[\sigma^{\left(  1-\beta\right)  /2}
\mathcal{N}^{\dag}\left(  \left[  \mathcal{N}\left(  \sigma\right)  \right]
^{\left(  \beta-1\right)  /2}\left[  \mathcal{N}\left(  \rho\right)  \right]
^{1-\beta}\left[  \mathcal{N}\left(  \sigma\right)  \right]  ^{\left(
\beta-1\right)  /2}\right)  \sigma^{\left(  1-\beta\right)  /2}\right]^{1/(1-\beta)}.
\end{multline}
Then the monotonicity follows from the ordinary monotonicity of the R\'{e}nyi
relative entropy, namely the statement that $D_{\alpha}\left(  \rho\Vert\sigma\right) \leq D_{\beta}\left(  \rho
\Vert\sigma\right)$ for $0\leq\alpha\leq\beta$~\cite{TCR09}. By a similar line of reasoning, when $1/\alpha+1/\beta=2$, $0\leq\alpha\leq\beta$, and $\mathcal{N}$ is a unitary quantum channel, the inequality in (\ref{monreldif2}) for $\widetilde{\Delta}_{\alpha}(\rho,\sigma,\mathcal{N})$ holds. This is because in such a case, we have that
$\alpha/\left(  1-\alpha\right)  =-\beta/\left(  1-\beta\right)  $, so that
\begin{multline}
\left[\sigma^{\left(
1-\alpha\right)  /2\alpha}\mathcal{N}^{\dag}\left(  \left[  \mathcal{N}\left(
\sigma\right)  \right]  ^{\left(  \alpha-1\right)  /2\alpha}\left[
\mathcal{N}\left(  \rho\right)  \right]  ^{\left(  1-\alpha\right)  /\alpha
}\left[  \mathcal{N}\left(  \sigma\right)  \right]  ^{\left(  \alpha-1\right)
/2\alpha}\right)  \sigma^{\left(  1-\alpha\right)  /2\alpha}\right]  ^{\alpha/\left(  1-\alpha\right)  }\\
=\left[\sigma^{\left(
1-\beta\right)  /2\beta}\mathcal{N}^{\dag}\left(  \left[  \mathcal{N}\left(
\sigma\right)  \right]  ^{\left(  \beta-1\right)  /2\beta}\left[
\mathcal{N}\left(  \rho\right)  \right]  ^{\left(  1-\beta\right)  /\beta
}\left[  \mathcal{N}\left(  \sigma\right)  \right]  ^{\left( -1\right)
/2\beta}\right)  \sigma^{\left(  1-\beta\right)  /2\beta}\right]  ^{\beta/\left(  1-\beta\right)  }.
\end{multline}
Since this is the case, the monotonicity follows directly from the ordinary monotonicity
of the sandwiched R\'{e}nyi relative entropy, namely, $\widetilde{D}_{\alpha}\left(  \rho\Vert\sigma\right)   \leq\widetilde
{D}_{\beta}\left(  \rho\Vert\sigma\right)$ for $0\leq\alpha\leq\beta$~\cite{MDSFT13}.

We also tested the inequalities in (\ref{monreldif1}) and (\ref{monreldif2}) numerically for states $\rho$ and $\sigma$ of dimensions up to five, and channels $\mathcal{N}$ of input and output dimensions (not necessarily the same) up to five, chosen at random according to the Haar measure~\cite{CubittMatlab}. For $\alpha$ and $\beta$ chosen in the range $(0,1) \cup (1, 10]$ such that $\alpha\leq\beta$, we found that the inequalities in (\ref{monreldif1}) and (\ref{monreldif2}) were never falsified in 1000 numerical runs for each setting of the parameters $\alpha$ and $\beta$.

A consequence of the above conjecture is the following inequality (a similar conjecture can also be found in~\cite[Eq.~(4.7)]{Z14b}). This bound, if proven true, would give an operationally useful answer to a question posed in~\cite{Winterconj}.

\begin{consequence}\label{monorprop1}
Let $\rho,\sigma\in\mathcal{S}\left(  \mathcal{H}\right)
_{++}$, and $\mathcal{N}$ be a strict CPTP map. Let $\mathcal{T}$ denote the
Petz recovery channel, which is indeed a CPTP\ map for $\rho$, $\sigma$ and
$\mathcal{N}$ as given:
\begin{equation}
\mathcal{T}\left(  \omega\right)  \equiv\sigma^{1/2}\mathcal{N}^{\dag}\left(
\left[  \mathcal{N}\left(  \sigma\right)  \right]  ^{-1/2}\omega\left[
\mathcal{N}\left(  \sigma\right)  \right]  ^{-1/2}\right)  \sigma^{1/2}.
\end{equation}
Assuming Conjecture~\ref{conj:mono-rel-diff}, the relative entropy difference
of~\eqref{eq:rel-entropy-diff} obeys the following lower bounds:
\begin{align}\label{eq:zhang-conj}
D\left(  \rho\Vert\sigma\right)  -D\left(  \mathcal{N}\left(  \rho\right)
\Vert\mathcal{N}\left(  \sigma\right)  \right)\geq-\log F\left(
\rho,\mathcal{T}\left(  \mathcal{N}(\rho)\right)  \right).
%\geq\frac{1}{4}\left\Vert\rho-\mathcal{T}\left(  \mathcal{N}(\rho)\right)  \right\Vert _{1}^{2}.
\end{align}
\end{consequence}

\begin{proof}
Conjecture~\ref{conj:mono-rel-diff} directly implies the lower bound:
\begin{align}
\widetilde{\Delta}_{1}\left(  \rho,\sigma,\mathcal{N}\right)   &
\geq\widetilde{\Delta}_{1/2}\left(  \rho,\sigma,\mathcal{N}\right)  \\
\Leftrightarrow D\left(  \rho\Vert\sigma\right)  -D\left(  \mathcal{N}\left(
\rho\right)  \Vert\mathcal{N}\left(  \sigma\right)  \right)   &  \geq
-\log\text{Tr}\left\{  \sqrt{\rho^{1/2}\mathcal{T}\left(  \mathcal{N}\left(
\rho\right)  \right)  \rho^{1/2}}\right\}  \label{forlinzconj}\\
&  =-\log F\left(  \rho,\mathcal{T}\left(  \mathcal{N}\left(  \rho\right)
\right)  \right)  .
\end{align}
%Next, consider that
%\begin{align}
%-\log F\left(  \rho,\mathcal{T}\left(  \mathcal{N}\left(  \rho\right)
%\right)  \right)   &  \geq-\log\left[  1-\left(  \frac{1}{2}\left\Vert
%\rho_{ABC}-\mathcal{T}\left(  \mathcal{N}\left(  \rho\right)  \right)
%\right\Vert _{1}\right)  ^{2}\right]  \\
%&  \geq\frac{1}{4}\left\Vert \rho_{ABC}-\mathcal{T}\left(  \mathcal{N}\left(
%\rho\right)  \right)  \right\Vert _{1}^{2},
%\end{align}
%In the above, the first inequality is a result of well known relations between
%trace distance and fidelity \cite{FG98}, and the second is a consequence of
%the inequality $-\log\left(  1-x\right)  \geq x$, which holds if $x\leq1$.
\end{proof}

%%%%%%%%%%%%%%%%%%%%%%%%%%%%%%%%%%%%%%%%%%%%%%%%%%%%%%%%%%%%%%%%%%%%%%%%%%%%%%%%%

\subsection{Conjectured remainder term for joint convexity of relative
entropy}

We now use (\ref{eq:zhang-conj}) to suggest a remainder term for joint
convexity of relative entropy. We once again emphasize that this remainder
term is a consequence of Conjecture~\ref{conj:mono-rel-diff}. Interestingly,
this inequality is equivalent to the inequality in (\ref{eq:zhang-conj}).

\begin{consequence}\label{rtjtconv}
Let $p_{X}\left(  x\right)  $ be a probability distribution characterizing the ensembles $\left\{  p_{X}\left(  x\right),\rho_{x}\right\}  $ and $\left\{  p_{X}\left(  x\right)  ,\sigma_{x}\right\}$ with $\rho_{x},\sigma_{x}\in\mathcal{S}\left(  \mathcal{H}\right)  _{++}$. From Consequence~\ref{monorprop1} we have the following remainder term for joint convexity of relative entropy:
\begin{equation}
\sum_{x}p_{X}\left(  x\right)  D\left(  \rho_{x}\Vert\sigma_{x}\right)
-D\left(  \overline{\rho}\Vert\overline{\sigma}\right)  \geq
-2\log\left(
\sum_{x}p_{X}\left(  x\right)\sqrt{F\left(\rho_{x},\sigma_{x}^{1/2}\left(
\overline{\sigma}\right)  ^{-1/2}\overline{\rho}\left(  \overline{\sigma
}\right)  ^{-1/2}\sigma_{x}^{1/2}\right)}\right)
,\label{jtconvrt}
\end{equation}
where $\overline{\rho}=\sum_{x}p_{X}\left(  x\right)  \rho_{x}$ and $
\overline{\sigma}=\sum_{x}p_{X}\left(  x\right)  \sigma_{x}.$ Furthermore, the inequality~\eqref{jtconvrt}
implies and is implied by the inequality~\eqref{eq:zhang-conj} (Consequence~\ref{monorprop1}),
so that these inequalities are equivalent.
\end{consequence}

\begin{proof}
Let
\begin{align}
\rho_{XB} &  \equiv\sum_{x}p_{X}\left(  x\right)  \left\vert x\right\rangle
\left\langle x\right\vert \otimes\rho_{x},\\
\sigma_{XB} &  \equiv\sum_{x}p_{X}\left(  x\right)  \left\vert x\right\rangle
\left\langle x\right\vert \otimes\sigma_{x},\\
\overline{\rho} &  \equiv\rho_{B}=\operatorname{Tr}_{X}\left\{  \rho
_{XB}\right\}  =\sum_{x}p_{X}\left(  x\right)  \rho_{x},\\
\overline{\sigma} &  \equiv\sigma_{B}=\operatorname{Tr}_{X}\left\{
\sigma_{XB}\right\}  =\sum_{x}p_{X}\left(  x\right)  \sigma_{x}.
\end{align}
Consider that (\ref{eq:zhang-conj}) with the map $\mathcal{N}$ being a
trace-out map over system $X$ implies that
\begin{align}
\sum_{x}p\left(  x\right)  D\left(  \rho_{x}\Vert\sigma_{x}\right)  -D\left(
\overline{\rho}\Vert\overline{\sigma}\right)   &  =D\left(  \rho_{XB}
\Vert\sigma_{XB}\right)  -D\left(  \rho_{B}\Vert\sigma_{B}\right)  \\
&\geq-\log F\left(\rho_{XB},\sigma_{XB}^{1/2}\sigma_{B}^{-1/2}
\rho_{B}\sigma_{B}^{-1/2}\sigma_{XB}^{1/2}\right)\\
%&  =\frac{1}{4}\left\Vert \sum_{x}p\left(  x\right)  \left\vert x\right\rangle
%\left\langle x\right\vert \otimes\left(  \rho_{x}-\sigma_{x}^{1/2}\left(
%\overline{\sigma}\right)  ^{-1/2}\overline{\rho}\left(  \overline{\sigma
%}\right)  ^{-1/2}\sigma_{x}^{1/2}\right)  \right\Vert _{1}^{2}\\
&  =-2\log\left(\sum_{x}p\left(  x\right)  \sqrt{F\left(\rho_{x},\sigma_{x}^{1/2}\left(  \overline{\sigma}\right)  ^{-1/2}\overline{\rho}\left(  \overline{\sigma}\right)  ^{-1/2}\sigma_{x}^{1/2}\right)}\right).
\end{align}
This establishes the first claim.
%The implication \eqref{eq:zhang-conj} $\Rightarrow$ \eqref{jtconvrt} already was used to prove Consequence~\ref{rtjtconv}. So we now consider the other implication.
Now suppose that
\begin{equation}
\sum_{x}p\left(  x\right)  D\left(  \rho_{x}\Vert\sigma_{x}\right)  -D\left(
\overline{\rho}\Vert\overline{\sigma}\right)\geq-2\log\left(\sum_{x}p\left(  x\right)  \sqrt{F\left(\rho_{x},\sigma_{x}^{1/2}\left(  \overline{\sigma}\right)  ^{-1/2}\overline{\rho}\left(  \overline{\sigma}\right)  ^{-1/2}\sigma_{x}^{1/2}\right)}\right)
\end{equation}
is true. We pick
\begin{align}
\rho_{x} &  =V_{E}^{i}U_{AE}\left(  \omega_{A}\otimes\left\vert 0\right\rangle
\left\langle 0\right\vert _{E}\right)  U_{AE}^{\dag}\left(  V_{E}^{i}\right)
^{\dag},\\
\sigma_{x} &  =V_{E}^{i}U_{AE}\left(  \tau_{A}\otimes\left\vert 0\right\rangle
\left\langle 0\right\vert _{E}\right)  U_{AE}^{\dag}\left(  V_{E}^{i}\right)
^{\dag},\\
p\left(  x\right)   &  =\frac{1}{d_{E}^{2}},
\end{align}
with $\left\{  V_{E}^{i}\right\}  _{i=1}^{d_{E}^{2}}$ a Heisenberg-Weyl set of
unitaries. These choices imply that
\begin{align}
\overline{\rho} =\mathcal{N}\left(  \omega_{A}\right)  \otimes\pi_{E},\quad\overline{\sigma} =\mathcal{N}\left(  \tau_{A}\right)  \otimes\pi_{E},
\end{align}
where $\pi_{E}$ is the maximally mixed state on system $E$, and $\mathcal{N}$
is defined by
\begin{align}
\mathcal{N}\left(  \cdot\right)=\text{Tr}_{E}\left\{  U_{AE}\left[
\left(  \cdot\right)  \otimes\left\vert 0\right\rangle \left\langle
0\right\vert _{E}\right]  U_{AE}^{\dag}\right\}=\sum_{j}\left\langle j\right\vert _{E}U_{AE}\left\vert 0\right\rangle
_{E}\left(  \cdot\right)  \left\langle 0\right\vert _{E}U_{AE}^{\dag
}\left\vert j\right\rangle _{E},
\end{align}
so that the Kraus operators of $\mathcal{N}$ are $\left\{  \left\langle
j\right\vert _{E}U_{AE}\left\vert 0\right\rangle _{E}\right\}  _{j}$. Then the
adjoint map $\mathcal{N}^{\dag}$ is given by
\begin{equation}
\mathcal{N}^{\dag}\left(  \cdot\right)  =\sum_{j}\left\langle 0\right\vert
_{E}U_{AE}^{\dag}\left\vert j\right\rangle _{E}\left(  \cdot\right)
\left\langle j\right\vert _{E}U_{AE}\left\vert 0\right\rangle _{E}.
\end{equation}
Furthermore, we have that
\begin{align}
\sum_{x}p\left(  x\right)  D\left(  \rho_{x}\Vert\sigma_{x}\right)   &
=D\left(  \omega_{A}\Vert\tau_{A}\right)  ,\\
D\left(  \overline{\rho}\Vert\overline{\sigma}\right)   &  =D\left(
\mathcal{N}\left(  \omega_{A}\right)  \Vert\mathcal{N}\left(  \tau_{A}\right)
\right)  .
\end{align}
Consider that
\begin{align}
&  \sigma_{x}^{1/2}\left(  \overline{\sigma}\right)  ^{-1/2}\overline{\rho
}\left(  \overline{\sigma}\right)  ^{-1/2}\sigma_{x}^{1/2}\nonumber\\
&  =\left[  V_{E}^{i}U_{AE}\left(  \tau_{A}\otimes\left\vert 0\right\rangle
\left\langle 0\right\vert _{E}\right)  U_{AE}^{\dag}\left(  V_{E}^{i}\right)
^{\dag}\right]  ^{1/2}\left[  \mathcal{N}\left(  \tau_{A}\right)  \otimes
\pi_{E}\right]  ^{-1/2}\left[  \mathcal{N}\left(  \omega_{A}\right)
\otimes\pi_{E}\right]  \times\nonumber\\
&  \ \ \ \ \ \ \ \ \left[  \mathcal{N}\left(  \tau_{A}\right)  \otimes\pi
_{E}\right]  ^{-1/2}\left[  V_{E}^{i}U_{AE}\left(  \tau_{A}\otimes\left\vert
0\right\rangle \left\langle 0\right\vert _{E}\right)  U_{AE}^{\dag}\left(
V_{E}^{i}\right)  ^{\dag}\right]  ^{1/2}\\
&  =V_{E}^{i}U_{AE}\left(  \tau_{A}^{1/2}\otimes\left\vert 0\right\rangle
\left\langle 0\right\vert _{E}\right)  U_{AE}^{\dag}\left(  V_{E}^{i}\right)
^{\dag}\left(  \left[  \mathcal{N}\left(  \tau_{A}\right)  \right]
^{-1/2}\mathcal{N}\left(  \omega_{A}\right)  \left[  \mathcal{N}\left(
\tau_{A}\right)  \right]  ^{-1/2}\otimes I_{E}\right)  \times\nonumber\\
&  \ \ \ \ \ \ \ \ V_{E}^{i}U_{AE}\left(  \tau_{A}^{1/2}\otimes\left\vert
0\right\rangle \left\langle 0\right\vert _{E}\right)  U_{AE}^{\dag}\left(
V_{E}^{i}\right)  ^{\dag}\\
&  =V_{E}^{i}U_{AE}\left(  \tau_{A}^{1/2}\otimes\left\vert 0\right\rangle
\left\langle 0\right\vert _{E}\right)  U_{AE}^{\dag}\left(  \left[
\mathcal{N}\left(  \tau_{A}\right)  \right]  ^{-1/2}\mathcal{N}\left(
\omega_{A}\right)  \left[  \mathcal{N}\left(  \tau_{A}\right)  \right]
^{-1/2}\otimes I_{E}\right)  \times\nonumber\\
&  \ \ \ \ \ \ \ \ U_{AE}\left(  \tau_{A}^{1/2}\otimes\left\vert
0\right\rangle \left\langle 0\right\vert _{E}\right)  U_{AE}^{\dag}\left(
V_{E}^{i}\right)  ^{\dag}
\end{align}
\begin{align}
&  =V_{E}^{i}U_{AE}\left(  \tau_{A}^{1/2}\otimes\left\vert 0\right\rangle
\left\langle 0\right\vert _{E}\right)  U_{AE}^{\dag}\left(  \left[
\mathcal{N}\left(  \tau_{A}\right)  \right]  ^{-1/2}\mathcal{N}\left(
\omega_{A}\right)  \left[  \mathcal{N}\left(  \tau_{A}\right)  \right]
^{-1/2}\otimes\sum_{j}\left\vert j\right\rangle \left\langle j\right\vert
_{E}\right)  \times\nonumber\\
&  \ \ \ \ \ \ \ \ U_{AE}\left(  \tau_{A}^{1/2}\otimes\left\vert
0\right\rangle \left\langle 0\right\vert _{E}\right)  U_{AE}^{\dag}\left(
V_{E}^{i}\right)  ^{\dag}\\
&  =V_{E}^{i}U_{AE}\left(  \tau_{A}^{1/2}\otimes\left\vert 0\right\rangle
_{E}\right)  \sum_{j}\left\langle 0\right\vert _{E}U_{AE}^{\dag}\left\vert
j\right\rangle _{E}\left(  \left[  \mathcal{N}\left(  \tau_{A}\right)
\right]  ^{-1/2}\mathcal{N}\left(  \omega_{A}\right)  \left[  \mathcal{N}
\left(  \tau_{A}\right)  \right]  ^{-1/2}\right)  \left\langle j\right\vert
_{E}U_{AE}\left\vert 0\right\rangle _{E}\times\nonumber\\
&  \ \ \ \ \ \ \ \ \left(  \tau_{A}^{1/2}\otimes\left\langle 0\right\vert
_{E}\right)  U_{AE}^{\dag}\left(  V_{E}^{i}\right)  ^{\dag}\\
&  =V_{E}^{i}U_{AE}\left(  \tau_{A}^{1/2}\otimes\left\vert 0\right\rangle
_{E}\right)  \mathcal{N}^{\dag}\left(  \left[  \mathcal{N}\left(  \tau
_{A}\right)  \right]  ^{-1/2}\mathcal{N}\left(  \omega_{A}\right)  \left[
\mathcal{N}\left(  \tau_{A}\right)  \right]  ^{-1/2}\right)  \left(  \tau
_{A}^{1/2}\otimes\left\langle 0\right\vert _{E}\right)  U_{AE}^{\dag}\left(
V_{E}^{i}\right)  ^{\dag}\\
&  =V_{E}^{i}U_{AE}\left(  \left[  \tau_{A}^{1/2}\mathcal{N}^{\dag}\left(
\left[  \mathcal{N}\left(  \tau_{A}\right)  \right]  ^{-1/2}\mathcal{N}\left(
\omega_{A}\right)  \left[  \mathcal{N}\left(  \tau_{A}\right)  \right]
^{-1/2}\right)  \tau_{A}^{1/2}\right]  \otimes\left\vert 0\right\rangle
\left\langle 0\right\vert _{E}\right)  U_{AE}^{\dag}\left(  V_{E}^{i}\right)
^{\dag}
\end{align}
So two to the power of minus one half times the remainder term of joint convexity evaluates to
\begin{align}
& \sum_{x}p\left(  x\right)  \sqrt{F\left(\rho_{x},\sigma_{x}^{1/2}\left(
\overline{\sigma}\right)  ^{-1/2}\overline{\rho}\left(  \overline{\sigma
}\right)  ^{-1/2}\sigma_{x}^{1/2}\right)}\nonumber\\
&=\sum_{i}\frac{1}{d_{E}^{2}}F^{1/2}\Bigg(
V_{E}^{i}U_{AE}\left(  \omega_{A}\otimes\left\vert 0\right\rangle \left\langle
0\right\vert _{E}\right)  U_{AE}^{\dag}\left(  V_{E}^{i}\right)  ^{\dag},\nonumber\\
&V_{E}^{i}U_{AE}\left(  \left[  \tau_{A}^{1/2}\mathcal{N}^{\dag}\left(  \left[
\mathcal{N}\left(  \tau_{A}\right)  \right]  ^{-1/2}\mathcal{N}\left(
\omega_{A}\right)  \left[  \mathcal{N}\left(  \tau_{A}\right)  \right]
^{-1/2}\right)  \tau_{A}^{1/2}\right]  \otimes\left\vert 0\right\rangle
\left\langle 0\right\vert _{E}\right)  U_{AE}^{\dag}\left(  V_{E}^{i}\right)
^{\dag}
\Bigg)\\
%F^{1/2}\left(
%\begin{array}
%[c]{c}
%V_{E}^{i}U_{AE}\left(  \omega_{A}\otimes\left\vert 0\right\rangle \left\langle
%0\right\vert _{E}\right)  U_{AE}^{\dag}\left(  V_{E}^{i}\right)  ^{\dag},\\
%V_{E}^{i}U_{AE}\left(  \left[  \tau_{A}^{1/2}\mathcal{N}^{\dag}\left(  \left[
%\mathcal{N}\left(  \tau_{A}\right)  \right]  ^{-1/2}\mathcal{N}\left(
%\omega_{A}\right)  \left[  \mathcal{N}\left(  \tau_{A}\right)  \right]
%^{-1/2}\right)  \tau_{A}^{1/2}\right]  \otimes\left\vert 0\right\rangle
%\left\langle 0\right\vert _{E}\right)  U_{AE}^{\dag}\left(  V_{E}^{i}\right)
%^{\dag}
%\end{array}
%\right)
&  =\sum_{i}\frac{1}{d_{E}^{2}}\sqrt{F\left(\omega_{A}\otimes\left\vert
0\right\rangle \left\langle 0\right\vert _{E},\left[  \tau_{A}^{1/2}
\mathcal{N}^{\dag}\left(  \left[  \mathcal{N}\left(  \tau_{A}\right)  \right]
^{-1/2}\mathcal{N}\left(  \omega_{A}\right)  \left[  \mathcal{N}\left(
\tau_{A}\right)  \right]  ^{-1/2}\right)  \tau_{A}^{1/2}\right]
\otimes\left\vert 0\right\rangle \left\langle 0\right\vert _{E}\right)}\\
&  =\sqrt{F\left(\omega_{A},\tau_{A}^{1/2}\mathcal{N}^{\dag}\left(  \left[
\mathcal{N}\left(  \tau_{A}\right)  \right]  ^{-1/2}\mathcal{N}\left(
\omega_{A}\right)  \left[  \mathcal{N}\left(  \tau_{A}\right)  \right]
^{-1/2}\right)  \tau_{A}^{1/2}\right)},
\end{align}
which gives~\eqref{eq:zhang-conj}.
\end{proof}

%\begin{proposition}
%The remainder term in \eqref{eq:zhang-conj} of Consequence~\ref{monorprop1}
%implies the remainder term in \eqref{jtconvrt} of Consequence~\ref{rtjtconv}, and vice versa.
%\end{proposition}

%%%%%%%%%%%%%%%%%%%%%%%%%%%%%%%%%%%%%%%%%%%%%%%%%%%%%%%%%%%%%%%%%%%%%%%%%%%%%%%%%

\subsection{A conjectured remainder term for the Holevo bound}

The Holevo bound~\cite{Ho73} is the statement that the following inequality holds
\begin{equation}
I\left(  X;B\right)  _{\rho}\geq I\left(  X;Y\right)  _{\omega}
,\label{eq:holevo-bound}
\end{equation}
where
\begin{align}
\rho_{XB}\equiv\sum_{x}p_{X}\left(  x\right)  \left\vert x\right\rangle
\left\langle x\right\vert _{X}\otimes\rho_{B}^{x},\quad\omega_{XY}  & =\mathcal{M}_{B\rightarrow Y}\left(  \rho_{XB}\right),
\end{align}
and $\mathcal{M}_{B\rightarrow Y}$ is a measurement channel defined as
\begin{align}
\mathcal{M}_{B\rightarrow Y}\left(  \sigma_{B}\right)\equiv\sum
_{y}\left\langle \varphi^{y}\right\vert \sigma_{B}\left\vert \varphi
^{y}\right\rangle \ \left\vert y\right\rangle \left\langle y\right\vert
_{Y},\quad\sum_{y}\left\vert \varphi^{y}\right\rangle \left\langle \varphi
^{y}\right\vert _{B}=I_{B}.
\end{align}
The Kraus operators for the measurement channel $\mathcal{M}_{B\rightarrow Y}$ are
$\left\{  \left\vert y\right\rangle \left\langle \varphi^{y}\right\vert
\right\}  $. (One might as well restrict to rank-one measurements since these
already give a tighter lower bound in (\ref{eq:holevo-bound}) than arbitrary measurements.)
We now write down a lower bound on the Holevo bound based on Consequence~\ref{monorprop1}.
The interpretation of the lower bound is that it quantifies how far away the states $\rho^x$ are from commuting. If the difference between the Holevo information and the classical mutual information is equal to zero, then the bound below holds with equality. In this case, each state $\rho^x$ is a fixed point of the same entanglement-breaking channel, which in turn by \cite[Theorem 5.3]{FNW14} implies that the states are commuting (a result well known since Holevo's original work).

\begin{consequence}
\label{cons:holevo-remainder}
The following remainder term for the Holevo bound is a consequence of the
conjecture in~\eqref{eq:zhang-conj}:
\begin{equation}
I\left(  X;B\right)  _{\rho}-I\left(  X;Y\right)  _{\omega}\geq-2\log\left(\sum
_{x}p_{X}\left(  x\right)  \sqrt{F\left(  \rho_{B}^{x},\mathcal{E}_B(\rho_{B}^{x})\right)
}\right),
\end{equation}
where $\mathcal{E}_B$ is an entanglement-breaking channel defined as
\begin{equation}
\mathcal{E}_B(\sigma_B) =\sum_{y}\left\langle \varphi^{y}\right\vert
\sigma_{B}\left\vert \varphi^{y}\right\rangle \frac{\rho_{B}^{1/2}\left\vert
\varphi^{y}\right\rangle \left\langle \varphi^{y}\right\vert _{B}\rho
_{B}^{1/2}}{\left\langle \varphi^{y}\right\vert
\rho_{B}\left\vert \varphi^{y}\right\rangle }.
\end{equation}
\end{consequence}

\begin{proof}
We would like to determine remainder terms for the inequality in
(\ref{eq:holevo-bound}) and can do so by rewriting it in terms of relative
entropies:
\begin{equation}
D\left(  \rho_{XB}\Vert\rho_{X}\otimes\rho_{B}\right)  -D\left(  \omega
_{XY}\Vert\omega_{X}\otimes\omega_{Y}\right)  \geq0.
\end{equation}
We now apply the conjecture in~\eqref{eq:zhang-conj} to the case of the
Holevo bound with
\begin{align}
\rho=\rho_{XB},\quad\sigma=\rho_{X}\otimes\rho_{B},\quad\mathcal{N}=\text{id}_{X}\otimes\mathcal{M}_{B\rightarrow Y}.
\end{align}
Consider for these choices that
\begin{align}
\mathcal{N}\left(  \rho\right)& =\omega_{XY}=\sum_{x,y}p_{X}\left(  x\right)  \left\langle \varphi^{y}\right\vert\rho_{B}^{x}\left\vert \varphi^{y}\right\rangle \ \left\vert x\right\rangle\left\langle x\right\vert _{X}\otimes\left\vert y\right\rangle \left\langle
y\right\vert _{Y},\\
\left[  \mathcal{N}\left(  \sigma\right)  \right]  ^{-1/2}  & =\left[
\rho_{X}\otimes\mathcal{M}_{B\rightarrow Y}\left(  \rho_{B}\right)  \right]
^{-1/2}\\
& =\rho_{X}^{-1/2}\otimes\left[  \mathcal{M}_{B\rightarrow Y}\left(  \rho
_{B}\right)  \right]  ^{-1/2}\\
& =\sum_{x}\left[  p_{X}\left(  x\right)  \right]  ^{-1/2}\left\vert
x\right\rangle \left\langle x\right\vert _{X}\otimes\sum_{y}\left[
\left\langle \varphi^{y}\right\vert \rho_{B}\left\vert \varphi^{y}
\right\rangle \right]  ^{-1/2}\ \left\vert y\right\rangle \left\langle
y\right\vert _{Y},\\
\mathcal{M}_{B\rightarrow Y}^{\dag}\left(  \cdot\right)    & =\sum
_{y}\left\vert \varphi^{y}\right\rangle \left\langle y\right\vert \left(
\cdot\right)  \left\vert y\right\rangle \left\langle \varphi^{y}\right\vert
,\\
\sigma^{1/2}  & =\rho_{X}^{1/2}\otimes\rho_{B}^{1/2}\\
& =\sum_{x}\left[  p_{X}\left(  x\right)  \right]  ^{1/2}\left\vert
x\right\rangle \left\langle x\right\vert _{X}\otimes\rho_{B}^{1/2}.
\end{align}
So we find that
\begin{equation}
\left[  \mathcal{N}\left(  \sigma\right)  \right]  ^{-1/2}\mathcal{N}\left(
\rho\right)  \left[  \mathcal{N}\left(  \sigma\right)  \right]  ^{-1/2}
=\sum_{x,y}\left\vert x\right\rangle \left\langle x\right\vert _{X}
\otimes\frac{\left\langle \varphi^{y}\right\vert \rho_{B}^{x}\left\vert
\varphi^{y}\right\rangle }{\left\langle \varphi^{y}\right\vert \rho
_{B}\left\vert \varphi^{y}\right\rangle }\left\vert y\right\rangle
\left\langle y\right\vert _{Y},
\end{equation}
and then
\begin{align}
\mathcal{N}^{\dag}\left(  \left[  \mathcal{N}\left(  \sigma\right)  \right]
^{-1/2}\mathcal{N}\left(  \rho\right)  \left[  \mathcal{N}\left(
\sigma\right)  \right]  ^{-1/2}\right)    & =\mathcal{M}_{B\rightarrow
Y}^{\dag}\left(  \sum_{x,y}\left\vert x\right\rangle \left\langle x\right\vert
_{X}\otimes\frac{\left\langle \varphi^{y}\right\vert \rho_{B}^{x}\left\vert
\varphi^{y}\right\rangle }{\left\langle \varphi^{y}\right\vert \rho
_{B}\left\vert \varphi^{y}\right\rangle }\left\vert y\right\rangle
\left\langle y\right\vert _{Y}\right)  \\
& =\sum_{x,y}\left\vert x\right\rangle \left\langle x\right\vert _{X}
\otimes\frac{\left\langle \varphi^{y}\right\vert \rho_{B}^{x}\left\vert
\varphi^{y}\right\rangle }{\left\langle \varphi^{y}\right\vert \rho
_{B}\left\vert \varphi^{y}\right\rangle }\left\vert \varphi^{y}\right\rangle
\left\langle \varphi^{y}\right\vert _{Y}.
\end{align}
Finally, we have that
\begin{align}
& \sigma^{1/2}\mathcal{N}^{\dag}\left(  \left[  \mathcal{N}\left(
\sigma\right)  \right]  ^{-1/2}\mathcal{N}\left(  \rho\right)  \left[
\mathcal{N}\left(  \sigma\right)  \right]  ^{-1/2}\right)  \sigma
^{1/2}\nonumber\\
& =\sum_{x,y}p_{X}\left(  x\right)  \left\vert x\right\rangle \left\langle
x\right\vert _{X}\otimes\frac{\left\langle \varphi^{y}\right\vert \rho_{B}
^{x}\left\vert \varphi^{y}\right\rangle }{\left\langle \varphi^{y}\right\vert
\rho_{B}\left\vert \varphi^{y}\right\rangle }\rho_{B}^{1/2}\left\vert
\varphi^{y}\right\rangle \left\langle \varphi^{y}\right\vert _{B}\rho
_{B}^{1/2}\\
& =\sum_{x,y}p_{X}\left(  x\right)  \frac{\left\langle \varphi^{y}\right\vert
\rho_{B}^{x}\left\vert \varphi^{y}\right\rangle }{\left\langle \varphi
^{y}\right\vert \rho_{B}\left\vert \varphi^{y}\right\rangle }\ \left\vert
x\right\rangle \left\langle x\right\vert _{X}\otimes\rho_{B}^{1/2}\left\vert
\varphi^{y}\right\rangle \left\langle \varphi^{y}\right\vert _{B}\rho
_{B}^{1/2}.
\end{align}
Defining
\begin{equation}
\sigma_{XB}\equiv\sum_{x}p_{X}\left(  x\right)  \left\vert x\right\rangle
\left\langle x\right\vert _{X}\otimes\sigma_{B}^{x}
\end{equation}
where
\begin{equation}
\sigma_{B}^{x}=\sum_{y}\frac{\left\langle \varphi^{y}\right\vert \rho_{B}
^{x}\left\vert \varphi^{y}\right\rangle }{\left\langle \varphi^{y}\right\vert
\rho_{B}\left\vert \varphi^{y}\right\rangle }\rho_{B}^{1/2}\left\vert
\varphi^{y}\right\rangle \left\langle \varphi^{y}\right\vert _{B}\rho
_{B}^{1/2},
\end{equation}
we see that
\begin{equation}
\sigma^{1/2}\mathcal{N}^{\dag}\left(  \left[  \mathcal{N}\left(
\sigma\right)  \right]  ^{-1/2}\mathcal{N}\left(  \rho\right)  \left[
\mathcal{N}\left(  \sigma\right)  \right]  ^{-1/2}\right)  \sigma^{1/2}
=\sigma_{XB}.
\end{equation}
So a conjectured remainder term for the Holevo bound is
\begin{align}
I\left(  X;B\right)  _{\rho}-I\left(  X;Y\right)  _{\omega}\geq-\log
F\left(  \rho_{XB},\sigma_{XB}\right)& =-2\log\sqrt{F\left(  \rho_{XB},\sigma_{XB}\right)  }\\
& =-2\log\left(\sum_{x}p_{X}\left(  x\right)  \sqrt{F\left(  \rho_{B}^{x},\sigma
_{B}^{x}\right)  }\right).
\end{align}
\end{proof}

%%%%%%%%%%%%%%%%%%%%%%%%%%%%%%%%%%%%%%%%%%%%%%%%%%%%%%%%%%%%%%%%%%%%%%%%%%%%%%%%%

\section{Discussion}

This paper has defined and examined various properties of a R\'{e}nyi squashed
entanglement and a R\'{e}nyi quantum discord. We took as a conjecture that the R\'{e}nyi
conditional mutual information of a tripartite state $\rho_{ABC}$ is monotone
with respect to local CPTP maps on both systems $A$ and $B$. Assuming the conjecture, we
showed that these quantities retain most of the properties of the original von
Neumann entropy-based quantities. For example, we showed that the R\'{e}nyi squashed
entanglement is convex, monotone with respect to LOCC, that it vanishes on separable
states and is subadditive on tensor-product states. Similarly, we showed that
the R\'{e}nyi quantum discord is non-negative, invariant with respect to the action of
local unitaries, vanishes on the set of classical-quantum states, and is
optimized by a rank-one POVM. Further, we proved relations of the R\'{e}nyi
squashed entanglement to a R\'{e}nyi entropy of entanglement and a R\'{e}nyi
entanglement of formation. We gave an expression for the R\'{e}nyi discord of pure
bipartite states. Also, assuming the truth of a conjecture on the monotonicity of the R\'{e}nyi conditional quantum mutual information with respect to the R\'{e}nyi parameter, we derived a remainder term for von Neumann entropy based quantum discord via the R\'{e}nyi quantum discord. 

This paper has also extended the procedure from our prior work for producing
R\'{e}nyi generalizations of quantum information measures to include measures
that are equal to differences of relative entropies. By conjecturing that the
proposed R\'{e}nyi generalizations of differences of relative entropies are
monotone in the R\'{e}nyi parameter, we suggested remainder terms for the
monotonicity and joint convexity of relative entropy, and the Holevo bound.
These remainder terms, if proven true, hold important applications in quantum
information theory~\cite{Winterconj, K13conj, Touchette14}.

The following are some future directions that could be considered based on the
results presented in this work. First of all, it remains a pressing open
question if the R\'{e}nyi conditional mutual information is monotone
with respect to local operations on both systems $A$ and $B$ (Conjecture~\ref{conj:monoA}). If Conjecture~\ref{conj:monoA} is false, then some of the consequences in this paper are directly impacted, such as the LOCC monotonicity of the R\'enyi squashed entanglement. One could also
try to prove more properties of the R\'{e}nyi squashed entanglement and
discord. For example, we have left open the converse part of faithfulness for both the R\'{e}nyi squashed entanglement as well as the R\'{e}nyi discord. The von Neumann entropy based squashed entanglement is known to be superadditive in general and additive on tensor-product states. However, we have only been able to show that the R\'{e}nyi squashed entanglement is subadditive on tensor-product states; super-additivity of the
R\'{e}nyi squashed entanglement in general has been left open. As far as
applications are concerned, it is an open question if the von
Neumann entropy based squashed entanglement is a strong converse rate for entanglement
distillation; the R{\'e}nyi squashed entanglement may be a useful tool in investigating this question. Also, using the R\'{e}nyi squashed entanglement, one could try
to prove that the von Neumann entropy based squashed entanglement is a strong converse rate
for the two-way assisted quantum capacity of any channel (the weak converse bound being shown in~\cite{TGW13, TGW14}). It might also be interesting to determine if
a Koashi-Winter type~\cite{KW04} relation holds for the proposed R\'{e}nyi discord.

\bigskip

\textbf{Note:} After the completion of this work, we learned of the recent 
result of \cite{FR14}, in which the following inequality was established for any tripartite state
$\rho_{ABC} \in\mathcal{S}(\mathcal{H}_A \otimes \mathcal{H}_B \otimes \mathcal{H}_C)$:
\begin{equation}
I(A;B|C)_\rho \geq -\log F(A;B|C)_\rho \, , \label{eq:FR}
\end{equation}
where $F(A;B|C)_\rho$ is the \textit{fidelity of recovery}, defined independently in
\cite{SW14} as
\begin{equation}
F(A;B|C)_\rho \equiv \sup_{\mathcal{R}} F(\rho_{ABC}, \mathcal{R}_{C\to AC} (\rho_{BC}) ) .
\end{equation}
The fidelity of recovery quantifies how well one can recover the full state on systems
$ABC$ if system $A$ is lost and one is allowed to perform a recovery map on system $C$ alone.
Thus, for states with small conditional mutual information (near to zero), the fidelity of recovery is high (near to one). Inequality \eqref{eq:FR} clearly impacts our Consequences
\ref{cons:discord} and \ref{cons:holevo-remainder}, so that one can interpret a state with small
quantum discord as being an approximate fixed point of an entanglement breaking channel
and an ensemble for which the
Holevo information is close to the accessible information as having all states in the ensemble
as approximate fixed points of the same entanglement breaking channel.

\bigskip

\textbf{Acknowledgments.} We thank A.~R.~P.~Rau for many helpful
discussions about this work and John Calsamiglia for discussions about
R\'{e}nyi quantum discord. KS and MMW are grateful to Naresh Sharma for
hosting them for a research visit to the Tata Institute of Fundamental
Research during June 2014, where some of the results in this paper were
established. MMW\ is grateful to the Institute for Quantum Information and
Matter at Caltech for hospitality during a research visit in July 2014. KS
acknowledges support from NSF\ Grant No.~CCF-1350397, the DARPA Quiness
Program through US Army Research Office award W31P4Q-12-1-0019, and the
Graduate School of Louisiana State University for the 2014-2015 Dissertation
Year Fellowship. MMW acknowledges support from the APS-IUSSTF Professorship
Awards in Physics, startup funds from the Department of Physics and Astronomy
at LSU, support from the NSF\ under Award No.~CCF-1350397, and support from
the DARPA Quiness Program through US Army Research Office award
W31P4Q-12-1-0019.

\appendix

%dummy comment inserted by tex2lyx to ensure that this paragraph is not empty%dummy comment inserted by tex2lyx to ensure that this paragraph is not empty

\section{R\'enyi conditional quantum mutual information: useful lemmas}

\subsection{Invariance with respect to tensoring with product states}\label{sub:Invaritens}

\begin{lemma}\label{lem:tensor-prod}
Let $\rho_{AA_{1}BB_{1}EE_{1}}\equiv\omega_{ABE}
\otimes\sigma_{A_{1}}\otimes\tau_{B_{1}}\otimes\gamma_{E_{1}}$. Then
\begin{equation}
I_{\alpha}\left(  AA_{1};BB_{1}|EE_{1}\right)  _{\rho}=I_{\alpha}\left(
A;B|E\right)  _{\omega}.
\end{equation}
\end{lemma}

\begin{proof}
This follows from a direct calculation. Consider that
\begin{align}
&  \text{Tr}_{AA_{1}}\left\{  \rho_{AA_{1}EE_{1}}^{\left(  1-\alpha\right)
/2}\rho_{AA_{1}BB_{1}EE_{1}}^{\alpha}\rho_{AA_{1}EE_{1}}^{\left(
1-\alpha\right)  /2}\right\} \nonumber\\
&  =\text{Tr}_{AA_{1}}\left\{  \left(  \omega_{AE}\otimes\sigma_{A_{1}}
\otimes\gamma_{E_{1}}\right)  ^{\left(  1-\alpha\right)  /2}\left(
\omega_{ABE}\otimes\sigma_{A_{1}}\otimes\tau_{B_{1}}\otimes\gamma_{E_{1}
}\right)  ^{\alpha}\left(  \omega_{AE}\otimes\sigma_{A_{1}}\otimes
\gamma_{E_{1}}\right)  ^{\left(  1-\alpha\right)  /2}\right\} \\
&  =\text{Tr}_{AA_{1}}\left\{  \left(  \omega_{AE}^{\left(  1-\alpha\right)
/2}\otimes\sigma_{A_{1}}^{\left(  1-\alpha\right)  /2}\otimes\gamma_{E_{1}
}^{\left(  1-\alpha\right)  /2}\right)  \left(  \omega_{ABE}^{\alpha}
\otimes\sigma_{A_{1}}^{\alpha}\otimes\tau_{B_{1}}^{\alpha}\otimes\gamma
_{E_{1}}^{\alpha}\right)  \left(  \omega_{AE}^{\left(  1-\alpha\right)
/2}\otimes\sigma_{A_{1}}^{\left(  1-\alpha\right)  /2}\otimes\gamma_{E_{1}
}^{\left(  1-\alpha\right)  /2}\right)  \right\} \\
&  =\text{Tr}_{AA_{1}}\left\{  \omega_{AE}^{\left(  1-\alpha\right)
/2}\left(  \omega_{ABE}^{\alpha}\otimes\tau_{B_{1}}^{\alpha}\right)
\omega_{AE}^{\left(  1-\alpha\right)  /2}\otimes\sigma_{A_{1}}\otimes
\gamma_{E_{1}}\right\} \\
&  =\text{Tr}_{A}\left\{  \omega_{AE}^{\left(  1-\alpha\right)  /2}
\omega_{ABE}^{\alpha}\omega_{AE}^{\left(  1-\alpha\right)  /2}\right\}
\otimes\tau_{B_{1}}^{\alpha}\otimes\gamma_{E_{1}}.
\end{align}
From the fact that
\begin{align}
&  \text{Tr}\left\{  \left(  \rho_{EE_{1}}^{\left(  \alpha-1\right)
/2}\text{Tr}_{AA_{1}}\left\{  \rho_{AA_{1}EE_{1}}^{\left(  1-\alpha\right)
/2}\rho_{AA_{1}BB_{1}EE_{1}}^{\alpha}\rho_{AA_{1}EE_{1}}^{\left(
1-\alpha\right)  /2}\right\}  \rho_{EE_{1}}^{\left(  \alpha-1\right)
/2}\right)  ^{1/\alpha}\right\} \nonumber\\
&  =\text{Tr}\left\{  \left(  \left(  \omega_{E}^{\left(  \alpha-1\right)
/2}\otimes\gamma_{E_{1}}^{\left(  \alpha-1\right)  /2}\right)  \left(
\text{Tr}_{A}\left\{  \omega_{AEE_{1}}^{\left(  1-\alpha\right)  /2}
\omega_{ABEE_{1}}^{\alpha}\omega_{AEE_{1}}^{\left(  1-\alpha\right)
/2}\right\}  \otimes\gamma_{E_{1}}\right)  \left(  \omega_{E}^{\left(
\alpha-1\right)  /2}\otimes\gamma_{E_{1}}^{\left(  \alpha-1\right)
/2}\right)  \otimes\tau_{B_{1}}^{\alpha}\right)  ^{1/\alpha}\right\} \\
&  =\text{Tr}\left\{  \left(  \omega_{E}^{\left(  \alpha-1\right)  /2}\left(
\text{Tr}_{A}\left\{  \omega_{AEE_{1}}^{\left(  1-\alpha\right)  /2}
\omega_{ABEE_{1}}^{\alpha}\omega_{AEE_{1}}^{\left(  1-\alpha\right)
/2}\right\}  \right)  \omega_{E}^{\left(  \alpha-1\right)  /2}\otimes
\tau_{B_{1}}^{\alpha}\otimes\gamma_{E_{1}}^{\alpha}\right)  ^{1/\alpha
}\right\} \\
&  =\text{Tr}\left\{  \left(  \omega_{E}^{\left(  \alpha-1\right)
/2}\text{Tr}_{A}\left\{  \omega_{AE}^{\left(  1-\alpha\right)  /2}\omega
_{ABE}^{\alpha}\omega_{AE}^{\left(  1-\alpha\right)  /2}\right\}  \omega
_{E}^{\left(  \alpha-1\right)  /2}\right)  ^{1/\alpha}\otimes\tau_{B_{1}
}\otimes\gamma_{E_{1}}\right\} \\
&  =\text{Tr}\left\{  \left(  \omega_{E}^{\left(  \alpha-1\right)
/2}\text{Tr}_{A}\left\{  \omega_{AE}^{\left(  1-\alpha\right)  /2}\omega
_{ABE}^{\alpha}\omega_{AE}^{\left(  1-\alpha\right)  /2}\right\}  \omega
_{E}^{\left(  \alpha-1\right)  /2}\right)  ^{1/\alpha}\right\},
\end{align}
it follows that
\begin{equation}
I_{\alpha}\left(  AA_{1};BB_{1}|EE_{1}\right)  _{\rho}=I_{\alpha}\left(
A;B|E\right)  _{\omega}.
\end{equation}
\end{proof}

%%%%%%%%%%%%%%%%%%%%%%%%%%%%%%%%%%%%%%%%%%%%%%%%%%%%%%%%%%%%%%%%%%%%%%%%%%%%%%%%%

\subsection{Conditioning on classical information}\label{sub:CononC}

\begin{lemma}\label{lem:classical-cond}
Let $\rho_{XABC}$ be a classical-quantum state of
the following form:
\begin{equation}
\rho_{XABC}=\sum_{x}p_{X}\left(  x\right)  \left\vert x\right\rangle
\left\langle x\right\vert _{X}\otimes\rho_{ABC}^{x}.
\end{equation}
Then the following identity holds for $\alpha > 0$:
\begin{equation}
I_{\alpha}\left(  A;B|CX\right)  _{\rho}=\frac{\alpha}{\alpha-1}\log\sum
_{x}p_{X}\left(  x\right)  \exp\left\{\left(  \frac{\alpha-1}{\alpha}\right)
I_{\alpha}\left(  A;B|C\right)  _{\rho^{x}}\right\}.
\end{equation}
\end{lemma}

\begin{proof}
Recalling the formula in (\ref{eq:CMI-sibson}), we have
\begin{equation}
I_{\alpha}\left(  A;B|CX\right)  _{\rho}=\frac{\alpha}{\alpha-1}\log
\text{Tr}\left\{  \left(  \text{Tr}_{A}\left\{  \rho_{CX}^{\left(
\alpha-1\right)  /2}\rho_{ACX}^{\left(  1-\alpha\right)  /2}\rho
_{ABCX}^{\alpha}\rho_{ACX}^{\left(  1-\alpha\right)  /2}\rho_{CX}^{\left(
\alpha-1\right)  /2}\right\}  \right)  ^{1/\alpha}\right\}.
\end{equation}
So
\begin{multline}
\rho_{CX}^{\left(  \alpha-1\right)  /2}\rho_{ACX}^{\left(  1-\alpha\right)
/2}\rho_{ABCX}^{\alpha}\rho_{ACX}^{\left(  1-\alpha\right)  /2}\rho
_{CX}^{\left(  \alpha-1\right)  /2}\\
=\left[  \sum_{x}p_{X}\left(  x\right)  \rho_{C}^{x}\otimes\left\vert
x\right\rangle \left\langle x\right\vert _{X}\right]  ^{\left(  \alpha
-1\right)  /2}\left[  \sum_{x}p_{X}\left(  x\right)  \rho_{AC}^{x}
\otimes\left\vert x\right\rangle \left\langle x\right\vert _{X}\right]
^{\left(  1-\alpha\right)  /2}\times\\
\left[  \sum_{x}p_{X}\left(  x\right)  \rho_{ABC}^{x}\otimes\left\vert
x\right\rangle \left\langle x\right\vert _{X}\right]  ^{\alpha}\left[
\sum_{x}p_{X}\left(  x\right)  \rho_{AC}^{x}\otimes\left\vert x\right\rangle
\left\langle x\right\vert _{X}\right]  ^{\left(  1-\alpha\right)  /2}\times\\
\left[  \sum_{x}p_{X}\left(  x\right)  \rho_{C}^{x}\otimes\left\vert
x\right\rangle \left\langle x\right\vert _{X}\right]  ^{\left(  \alpha
-1\right)  /2}
\end{multline}
\begin{equation}
=\sum_{x}p_{X}^{\alpha}\left(  x\right)  \left[  \rho_{C}^{x}\right]
^{\left(  \alpha-1\right)  /2}\left[  \rho_{AC}^{x}\right]  ^{\left(
1-\alpha\right)  /2}\left[  \rho_{ABC}^{x}\right]  ^{\alpha}\left[  \rho
_{AC}^{x}\right]  ^{\left(  1-\alpha\right)  /2}\left[  \rho_{C}^{x}\right]
^{\left(  \alpha-1\right)  /2}\otimes\left\vert x\right\rangle \left\langle
x\right\vert _{X}.
\end{equation}
From the fact that
\begin{align}
&  \text{Tr}\left\{  \left(  \text{Tr}_{A}\left\{  \sum_{x}p_{X}^{\alpha
}\left(  x\right)  \left[  \rho_{C}^{x}\right]  ^{\left(  \alpha-1\right)
/2}\left[  \rho_{AC}^{x}\right]  ^{\left(  1-\alpha\right)  /2}\left[
\rho_{ABC}^{x}\right]  ^{\alpha}\left[  \rho_{AC}^{x}\right]  ^{\left(
1-\alpha\right)  /2}\left[  \rho_{C}^{x}\right]  ^{\left(  \alpha-1\right)
/2}\otimes\left\vert x\right\rangle \left\langle x\right\vert _{X}\right\}
\right)  ^{1/\alpha}\right\} \nonumber\\
&  =\text{Tr}\left\{  \left(  \sum_{x}p_{X}^{\alpha}\left(  x\right)
\text{Tr}_{A}\left\{  \left[  \rho_{C}^{x}\right]  ^{\left(  \alpha-1\right)
/2}\left[  \rho_{AC}^{x}\right]  ^{\left(  1-\alpha\right)  /2}\left[
\rho_{ABC}^{x}\right]  ^{\alpha}\left[  \rho_{AC}^{x}\right]  ^{\left(
1-\alpha\right)  /2}\left[  \rho_{C}^{x}\right]  ^{\left(  \alpha-1\right)
/2}\right\}  \otimes\left\vert x\right\rangle \left\langle x\right\vert
_{X}\right)  ^{1/\alpha}\right\} \\
&  =\text{Tr}\left\{  \sum_{x}p_{X}\left(  x\right)  \left(  \text{Tr}
_{A}\left\{  \left[  \rho_{C}^{x}\right]  ^{\left(  \alpha-1\right)
/2}\left[  \rho_{AC}^{x}\right]  ^{\left(  1-\alpha\right)  /2}\left[
\rho_{ABC}^{x}\right]  ^{\alpha}\left[  \rho_{AC}^{x}\right]  ^{\left(
1-\alpha\right)  /2}\left[  \rho_{C}^{x}\right]  ^{\left(  \alpha-1\right)
/2}\right\}  \right)  ^{1/\alpha}\otimes\left\vert x\right\rangle \left\langle
x\right\vert _{X}\right\} \\
&  =\sum_{x}p_{X}\left(  x\right)  \text{Tr}\left\{  \left(  \text{Tr}
_{A}\left\{  \left[  \rho_{C}^{x}\right]  ^{\left(  \alpha-1\right)
/2}\left[  \rho_{AC}^{x}\right]  ^{\left(  1-\alpha\right)  /2}\left[
\rho_{ABC}^{x}\right]  ^{\alpha}\left[  \rho_{AC}^{x}\right]  ^{\left(
1-\alpha\right)  /2}\left[  \rho_{C}^{x}\right]  ^{\left(  \alpha-1\right)
/2}\right\}  \right)  ^{1/\alpha}\right\} \\
&  =\sum_{x}p_{X}\left(  x\right)  \exp\left\{\left(  \frac{\alpha-1}{\alpha}\right)
I_{\alpha}\left(  A;B|C\right)  _{\rho^{x}}\right\},
\end{align}
it follows that
\begin{equation}
I_{\alpha}\left(  A;B|CX\right)  _{\rho}=\frac{\alpha}{\alpha-1}\log\sum
_{x}p_{X}\left(  x\right)  \exp\left\{\left(  \frac{\alpha-1}{\alpha}\right)
I_{\alpha}\left(  A;B|C\right)  _{\rho^{x}}\right\}.
\end{equation}
\end{proof}

%%%%%%%%%%%%%%%%%%%%%%%%%%%%%%%%%%%%%%%%%%%%%%%%%%%%%%%%%%%%%%%%%%%%%%%%%%%%%%%%%

\section{Proof of Conjecture~\ref{conj:mono-rel-diff} for $\alpha$ in a neighborhood of one}\label{renreldifApp}

This section gives more details about proving Conjecture~\ref{conj:mono-rel-diff} in a neighborhood of one. Consider the inequality in (\ref{monreldif1}). Let $\rho$, $\sigma\in\mathcal{S}\left(\mathcal{H}\right)_{++}$. We begin by introducing a variable
\begin{equation}
\gamma=\alpha-1,
\end{equation}
and with%
\begin{equation}
Y\left(  \gamma\right)  \equiv\rho^{1+ \gamma}\sigma^{- \gamma /2}
\mathcal{N}^{\dag}\left(  \left[  \mathcal{N}\left(  \sigma\right)  \right]
^{\gamma  /2}\left[  \mathcal{N}\left(  \rho\right)  \right]
^{- \gamma}\left[  \mathcal{N}\left(  \sigma\right)  \right]  ^{
 \gamma  /2}\right)  \sigma^{  - \gamma  /2},\label{eq:Y_beta}
\end{equation}
it follows that $\Delta_{\alpha}\left(  \rho,\sigma,\mathcal{N}\right)  $ is equal to
\begin{equation}
\frac{1}{\alpha-1}\log
\text{Tr}\left\{  \rho^{\alpha}\sigma^{\left(  1-\alpha\right)  /2}
\mathcal{N}^{\dag}\left(  \left[  \mathcal{N}\left(  \sigma\right)  \right]
^{\left(  \alpha-1\right)  /2}\left[  \mathcal{N}\left(  \rho\right)  \right]
^{1-\alpha}\left[  \mathcal{N}\left(  \sigma\right)  \right]  ^{\left(
\alpha-1\right)  /2}\right)  \sigma^{\left(  1-\alpha\right)  /2}\right\}
=\frac{1}{\gamma}\log\text{Tr}\left\{  Y\left(  \gamma\right)  \right\}  .
\end{equation}
Since $d\gamma/d\alpha=1$,
\begin{multline}
\frac{d}{d\alpha}\left[ \frac{1}{\alpha-1}\log
\text{Tr}\left\{  \rho^{\alpha}\sigma^{\left(  1-\alpha\right)  /2}
\mathcal{N}^{\dag}\left(  \left[  \mathcal{N}\left(  \sigma\right)  \right]
^{\left(  \alpha-1\right)  /2}\left[  \mathcal{N}\left(  \rho\right)  \right]
^{1-\alpha}\left[  \mathcal{N}\left(  \sigma\right)  \right]  ^{\left(
\alpha-1\right)  /2}\right)  \sigma^{\left(  1-\alpha\right)  /2}\right\} \right]  \\=\frac{d}{d\gamma}\left[  \frac{1}{\gamma
}\log\text{Tr}\left\{  Y\left(  \gamma\right)  \right\}  \right]  .
\end{multline}
We can then explicitly compute the derivative:
\begin{align}
\frac{d}{d\gamma}\left[  \frac{1}{\gamma}\log\text{Tr}\left\{  Y\left(
\gamma\right)  \right\}  \right]   &  =-\frac{1}{\gamma^{2}}\log
\text{Tr}\left\{  Y\left(  \gamma\right)  \right\}  +\frac{\text{Tr}\left\{
\frac{d}{d\gamma}Y\left(  \gamma\right)  \right\}  }{\gamma\text{Tr}\left\{
Y\left(  \gamma\right)  \right\}  }\\
&  =\frac{\gamma\text{Tr}\left\{  \frac{d}{d\gamma}Y\left(  \gamma\right)
\right\}  -\text{Tr}\left\{  Y\left(  \gamma\right)  \right\}  \log
\text{Tr}\left\{  Y\left(  \gamma\right)  \right\}  }{\gamma^{2}
\text{Tr}\left\{  Y\left(  \gamma\right)  \right\}  }.
\label{eq:derivative-exp}
\end{align}

We can prove that the numerator of
(\ref{eq:derivative-exp}) is non-negative for $\gamma$ in a neighborhood of
zero. To this end, consider a Taylor expansion of $Y\left(  \gamma\right)  $
in (\ref{eq:Y_beta})\ around $\gamma$ equal to zero (so around $\alpha$ equal
to one). Indeed, consider that
\begin{align}
X^{1+\gamma}  &  =X+\gamma X\log X+\frac{\gamma^{2}}{2}X\log^{2}X+O\left(
\gamma^{3}\right)  ,\\
X^{\gamma}  &  =I+\gamma\log X+\frac{\gamma^{2}}{2}\log^{2}X+O\left(
\gamma^{3}\right)  .
\end{align}
For our case, we make the following substitutions into Tr$\left\{  Y\left(
\gamma\right)  \right\}  $:
\begin{align}
\rho^{1+\gamma}  &  =\rho+\gamma\rho\log\rho
+\frac{\gamma^{2}}{2}\rho\log^{2}\rho+O\left(  \gamma^{3}\right)
,\\
\sigma^{- \gamma /2}&=I-\frac{\gamma}{2}\log\sigma
+\frac{\gamma^{2}}{8}\log^{2}\sigma+O\left(  \gamma^{3}\right),\\
\left[  \mathcal{N}\left(  \sigma\right)  \right]
^{\gamma  /2}  &  =I+\frac{\gamma}{2}\log\left[  \mathcal{N}\left(  \sigma\right)  \right]
+\frac{\gamma^{2}}{8}\log^{2}\left[  \mathcal{N}\left(  \sigma\right)  \right]+O\left(  \gamma^{3}\right)  ,\\
\left[  \mathcal{N}\left(  \rho\right)  \right]^{- \gamma}&=I-\gamma\log\left[  \mathcal{N}\left(  \rho\right)  \right]+\frac{\gamma^{2}}{2}\log
^{2}\left[  \mathcal{N}\left(  \rho\right)  \right]+O\left(  \gamma^{3}\right)  .
\end{align}
After a rather tedious calculation, we find that
\begin{equation}
\text{Tr}\left\{  Y\left(  \gamma\right)  \right\}  =\text{Tr}\left\{
\rho\right\}  +\gamma\Delta\left(  \rho,\sigma,\mathcal{N}\right)
+\frac{\gamma^{2}}{2}\left[  V\left(  \rho,\sigma,\mathcal{N}\right)  +\left[
\Delta\left( \rho,\sigma,\mathcal{N}\right)  \right]  ^{2}\right]  +O\left(
\gamma^{3}\right)  ,
\end{equation}
where $\Delta\left(  \rho,\sigma,\mathcal{N}\right)$ is the relative entropy difference (of the form written in (\ref{vonNreldiff})) and $V\left(  \rho,\sigma,\mathcal{N}\right)  $ is a relative entropy difference variance of $\rho$, $\sigma$, and $\mathcal{N}$:
\begin{multline}
 V\left(  \rho,\sigma,\mathcal{N}\right) \equiv \text{Tr}\left\{  \rho\left[  \log\rho-\log\sigma-\mathcal{N}^{\dag}\left(
\log\mathcal{N}\left(  \rho\right)  \right)  +\mathcal{N}^{\dag}\left(
\log\mathcal{N}\left(  \sigma\right)  \right) 
- \Delta\left(  \rho,\sigma,\mathcal{N}\right)
\right]  ^{2}\right\}  \\
  +\text{Tr}\left\{  \mathcal{N}\left(  \rho\right)  \left[  \log
\mathcal{N}\left(  \rho\right)  -\log\mathcal{N}\left(  \sigma\right)
\right]  ^{2}\right\}  -\text{Tr}\left\{  \rho\left[  \mathcal{N}^{\dag
}\left(  \log\mathcal{N}\left(  \rho\right)  \right)  -\mathcal{N}^{\dag
}\left(  \log\mathcal{N}\left(  \sigma\right)  \right)  \right]  ^{2}\right\}
.
\end{multline}
Letting $A=\log\mathcal{N}\left(  \rho\right)  -\log\mathcal{N}\left(
\sigma\right)  $, observe that%
\begin{align}
& \text{Tr}\left\{  \mathcal{N}\left(  \rho\right)  \left[  \log
\mathcal{N}\left(  \rho\right)  -\log\mathcal{N}\left(  \sigma\right)
\right]  ^{2}\right\}  -\text{Tr}\left\{  \rho\left[  \mathcal{N}^{\dag
}\left(  \log\mathcal{N}\left(  \rho\right)  \right)  -\mathcal{N}^{\dag
}\left(  \log\mathcal{N}\left(  \sigma\right)  \right)  \right]  ^{2}\right\} \nonumber
\\
& =\text{Tr}\left\{  \mathcal{N}\left(  \rho\right)  A^{\dag}A\right\}
-\text{Tr}\left\{  \rho\mathcal{N}^{\dag}\left(  A^{\dag}\right)
\mathcal{N}^{\dag}\left(  A\right)  \right\}  \\
& =\text{Tr}\left\{  \rho\mathcal{N}^{\dag}\left(  A^{\dag}A\right)  \right\}
-\text{Tr}\left\{  \rho\mathcal{N}^{\dag}\left(  A^{\dag}\right)
\mathcal{N}^{\dag}\left(  A\right)  \right\}  \\
& =\text{Tr}\left\{  \rho\left[  \mathcal{N}^{\dag}\left(  A^{\dag}A\right)
-\mathcal{N}^{\dag}\left(  A^{\dag}\right)  \mathcal{N}^{\dag}\left(
A\right)  \right]  \right\}  \\
& \geq0.
\end{align}
The inequality follows by applying the Kadison-Schwarz inequality \cite{B07}, i.e., for
any 2-positive unital map $\phi$, we have that%
\begin{equation}
\phi(  a^{\dag}a)  \geq\phi(  a^{\dag})  \phi\left(
a\right)  .
\end{equation}
For any Hermitian operator $H$, we have that
\begin{equation}
\left\langle H^{2}\right\rangle _{\rho}-\left\langle H\right\rangle _{\rho
}^{2}\geq0.
\end{equation}
So taking $H\equiv\log\rho-\log\sigma-\mathcal{N}^\dagger\log\mathcal{N}(\rho)+\mathcal{N}^\dagger\log\mathcal{N}(\sigma)
$, we conclude that
\begin{equation}\text{Tr}\left\{  \rho\left[  \log\rho-\log\sigma-\mathcal{N}^{\dag}\left(
\log\mathcal{N}\left(  \rho\right)  \right)  +\mathcal{N}^{\dag}\left(
\log\mathcal{N}\left(  \sigma\right)  \right) 
-\Delta\left(  \rho,\sigma,\mathcal{N}\right)
\right]  ^{2}\right\}\geq0 .
\end{equation}
Putting the two inequalities above together, we conclude that
$V\left(  \rho,\sigma,\mathcal{N}\right) \geq 0$, an
observation central to our development here.

We will make the abbreviations
$\Delta\equiv\Delta\left(  \rho,\sigma,\mathcal{N}\right)  $ and $V\equiv
V\left(  \rho,\sigma,\mathcal{N}\right)  $\ from here forward, so that
\begin{equation}
\text{Tr}\left\{  Y\left(  \gamma\right)  \right\}  =1+\gamma\Delta
+\frac{\gamma^{2}}{2}\left[  V+\Delta^{2}\right]  +O\left(  \gamma^{3}\right)
. \label{eq:beta-small-proof-1}
\end{equation}
So this implies that
\begin{align}
\gamma\text{Tr}\left\{  \frac{d}{d\gamma}Y\left(  \gamma\right)  \right\}&
=\gamma\Delta+\gamma^{2}\left[  V+\Delta^{2}\right]  +O\left(  \gamma
^{3}\right)  ,\\
\text{Tr}\left\{  Y\left(  \gamma\right)  \right\}  \log\text{Tr}\left\{
Y\left(  \gamma\right)  \right\}   &  =\left[  1+\gamma\Delta+\frac{\gamma
^{2}}{2}\left[  V+\Delta^{2}\right]  +O\left(  \gamma^{3}\right)  \right]
\log\left[  1+\gamma\Delta+\frac{\gamma^{2}}{2}\left[  V+\Delta^{2}\right]
+O\left(  \gamma^{3}\right)  \right]  .
\end{align}
Then for small $\gamma$, we have the following Taylor expansion for the
logarithm:
\begin{align}
\log\left[  1+\gamma\Delta+\frac{\gamma^{2}}{2}\left[  V+\Delta^{2}\right]
+O\left(  \gamma^{3}\right)  \right]   &  =\gamma\Delta+\frac{\gamma^{2}}
{2}\left[  V+\Delta^{2}\right]  -\frac{\gamma^{2}\Delta^{2}}{2}+O\left(
\gamma^{3}\right) \\
&  =\gamma\Delta+\frac{\gamma^{2}}{2}V+O\left(  \gamma^{3}\right)  ,
\end{align}
which gives
\begin{align}
\text{Tr}\left\{  Y\left(  \gamma\right)  \right\}  \log\text{Tr}\left\{
Y\left(  \gamma\right)  \right\}   &  =\left[  1+\gamma\Delta+\frac{\gamma
^{2}}{2}\left[  V+\Delta^{2}\right]  +O\left(  \gamma^{3}\right)  \right]
\left[  \gamma\Delta+\frac{\gamma^{2}}{2}V+O\left(  \gamma^{3}\right)  \right]
\nonumber\\
&  =\gamma\Delta+\frac{\gamma^{2}}{2}V+\gamma^{2}\Delta^{2}+O\left(
\gamma^{3}\right)  .
\end{align}
Finally, we can say that
\begin{align}
\gamma\text{Tr}\left\{  \frac{d}{d\gamma}Y\left(  \gamma\right)  \right\}
-\text{Tr}\left\{  Y\left(  \gamma\right)  \right\}  \log\text{Tr}\left\{
Y\left(  \gamma\right)  \right\}   &  =\gamma\Delta+\gamma^{2}\left[
V+\Delta^{2}\right]  -\left[  \gamma\Delta+\frac{\gamma^{2}}{2}V+\gamma
^{2}\Delta^{2}\right]  +O\left(  \gamma^{3}\right) \nonumber\\
&  =\frac{\gamma^{2}}{2}V+O\left(  \gamma^{3}\right)  .
\label{eq:beta-small-proof-last}
\end{align}
If $V>0$, we can conclude that as long as $\gamma$ is very near to zero, all
terms $O\left(  \gamma^{3}\right)  $ are negligible in comparison to
$\frac{\gamma^{2}}{2}V$, and the monotonicity holds in such a regime. (Note that this
argument does not work if $V=0$.)

A similar kind of development can also be used to show that the the inequality in (\ref{monreldif2}) holds for $\gamma$ in a neighborhood of zero (i.e., $\alpha$ in a neighborhood of one).

\bibliographystyle{alpha}
\bibliography{Ref-1}

\end{document}